\documentclass[a4paper,12pt]{article}
\usepackage[bottom=4.5cm,top=3.8cm,left=2.2cm,right=2.2cm]{geometry}
\let\OLDthebibliography\thebibliography
\renewcommand\thebibliography[1]{
  \OLDthebibliography{#1}
  \setlength{\parskip}{1pt}
  \setlength{\itemsep}{2pt plus 0.3ex}
}

\usepackage{graphicx}
\usepackage{caption}
\usepackage{subcaption}	
\usepackage[shortlabels]{enumitem}

\usepackage{amsthm}
\usepackage{amsmath,amssymb}
\usepackage{mathtools}
\usepackage{chngcntr}
\usepackage{url}

\usepackage{algorithmicx}
\usepackage{algpseudocode}
\usepackage{algorithm}

\newcommand{\gap}{{\mathop{\rm gap}}}

\newtheorem{prop}{Proposition}[section]
\newtheorem{lemma}{Lemma}[section]

\newtheorem{theorem}{Theorem}[section]
\newtheorem{corollary}{Corollary}[section]

\newtheorem{observation}{Observation}[section]

\newtheorem{property}{Property}[section]

\begin{document}

\begin{center}
{\Large Multimodal Transportation with Ridesharing of Personal Vehicles}
\vskip 0.2in
Qian-Ping Gu$^1$, Jiajian Leo Liang$^1$

$^1$School of Computing Science, Simon Fraser University, Canada\\
qgu@sfu.ca, leo\_liang@sfu.ca
\end{center}
\vspace{2mm}

\noindent \textbf{Abstract:}
Many public transportation systems are unable to keep up with growing passenger demand as the population grows in urban areas.
The slow or lack of improvements for public transportation pushes people to use private transportation modes, such as carpooling and ridesharing.
However, the occupancy rate of personal vehicles has been dropping in many cities.
In this paper, we describe a centralized transit system that integrates public transit and ridesharing, which matches drivers and transit riders such that the riders would result in shorter travel time using both transit and ridesharing. The optimization goal of the system is to assign as many riders to drivers as possible for ridesharing.
We give an exact approach and approximation algorithms to achieve the optimization goal.
As a case study, we conduct an extensive computational study to show the effectiveness of the transit system for different approximation algorithms, based on the real-world traffic data in Chicago City; the data sets include both public transit and ridesharing trip information.
The experiment results show that our system is able to assign more than 60\% of riders to drivers, leading to a substantial increase in occupancy rate of personal vehicles and reducing riders' travel time.

\vspace{2mm}

\noindent \textbf{Keywords:} Multimodal transportation, ridesharing, approximation algorithms, computational study

\section{Introduction} \label{sec-intro}
As the population grows in urban areas, commuting between and within large cities is time-consuming and resource-demanding.
Due to growing passenger demand, the number of vehicles on the road for both public and private transportation has increased to handle the demand.
Public transportation systems are unable to keep up with the demand in terms of service quality.
This pushes people to use personal vehicles for work commute.
In the United States, personal vehicles are the main transportation mode~\cite{CSS20}.
However, the occupancy rate of personal vehicles in the U.S. is 1.6 persons per vehicle in 2011~\cite{USDT-G,USDTFHA-S} (and decreased to 1.5 persons per vehicle in 2017~\cite{CSS20}), which can be a major cause for congestion and pollution.
This is the reason municipal governments encourage the use of public transit;
the major drawback of public transit is the inconvenience of last mile and/or first mile transportation compared to personal vehicles~\cite{TS14-W}.
With the increasing popularity in ridesharing/ridehailing service, there may be potential in integrating private and public transportation.
From the research report of~\cite{TRB16-M}, it is recommended that public transit agencies should build on mobility innovations to allow public-private engagement in ridesharing because the use of shared modes increases the likelihood of using public transit.
As pointed out by Ma et al.~\cite{TRELTR19-M}, some basic form of collaboration between MoD (mobility-on-demand) services and public transit already exists (for first and last mile transportation).
There is an increasing interest for collaboration between private companies and public sector entities~\cite{SMT18-R}.

The spareness of transit networks usually is the main cause of the inconvenience in public transit.
Such transit networks have infrequent transit schedule and cause customers to have multiple transfers.
In this paper, we investigate the potential effectiveness of integrating public transit with ridesharing to increase ridership in such sparse transit networks and reduce traffic congestion for work commute (not very short trips).
For example, people who drive their vehicles to work can pick-up \textit{riders}, who use public transit regularly, and drop-off them at some transit stops, and those riders can take public transit to their destinations.
In this way, riders are presented with a cheaper alternative than ridesharing for the entire trip, and it is more convenient than using public transit only.
The transit system also gets a higher ridership, which matches the recommendation of~\cite{TRB16-M} for a more sustainable transportation system.
Our research focuses on a centralized system that is capable of matching drivers and riders satisfying their trips' requirements while achieving some optimization goal; the requirements of a trip may include an origin and a destination, time constraints, capacity of a vehicle, and so on.
When a rider is assigned a driver, we call this \emph{ridesharing route}, and it is compared with the fastest \emph{public transit route} for this rider which uses only public transit. 
If the ridesharing route is faster than the public transit route, the ridesharing route is provided to both the rider and driver.
To increase the number of rider participants, our system-wide optimization goal is to maximize the number of riders, each of whom is assigned a ridesharing route. We call this the \emph{maximization problem} (formal definition in Section~\ref{sec-preliminary}).

In the literature, there are many papers about standalone ridesharing/carpooling, from theoretical to empirical studies (e.g.,~\cite{TRBM11-A,PNAS17-AM,GLZ21,TRBM20-X}).
For literature reviews on ridesharing, readers are referred to~\cite{EJOR12-A,TRBM13-F,TRBM19-M,SS20-T}.
On the other hand, there are only few papers study the integration of public transit with dynamic ridesharing.
Aissat and Varone~\cite{ICEIS15-A} proposed an approach in which a public transit route for each rider is given, and their algorithm tries to substitute part(s) of every rider's route with ridesharing.
Any part of a rider's original transit route is replaced only if ridesharing substitution is better than the original part.
Their algorithm finds the best route for each rider in first-come first-serve basis (system-wide optimization goal is not considered) and is computational intensive.
Huang et al.~\cite{TITS19-H} presented a more robust approach, compared to~\cite{ICEIS15-A}, 
by combining two networks $N, N'$ (representing the public transit and ridesharing network respectively) into one single routable graph $G$.
The graph $G$ uses the \emph{time-expanded model} to maintain the information about all public vehicles schedule, riders' and drivers' origins, destinations and time constraints.
In general, a \emph{stop node} in $G$ represents a public vehicle's/driver's stop location, and a \emph{time node} represents time events of this vehicle/driver at this stop.
An edge between two nodes implies possible transfer for riders from one vehicle to the other (i.e., the departure time of a vehicle is after the arrival of the other); this also implies a rider can be pick-up/dropped-off from/at a public stop within time constraints.
The authors apply this idea to create the ridesharing network graph $N'$ and connect the two networks $N, N'$ by creating edges between them whenever a rider can be pick-up/dropped-off from/at a public stop within time constraints.
For each rider travel query, a shortest path is found on $G$.
Their approach is also first-come first-serve basis and does not achieve system-wide optimization goal.

Ma~\cite{EEEIC17-M} and Stiglic et al.~\cite{COR18-S} proposed models to integrate ridesharing and public transit as graph matching problems to achieve system-wide optimization goals.
Algorithm presented in~\cite{EEEIC17-M} uses the shareability graph (RV-graph)~\cite{PNAS14-S} and the extension of RV-graph, called RTV-graph~\cite{PNAS17-AM}.
In fact, the approach used by Stiglic et al.~\cite{COR18-S} is similar, except~\cite{COR18-S} supports more rideshare match types.
A set of driver and rider trip announcements and a public transit network with a fixed cyclic timetable are given.
For a pre-transit rideshare match, a set of riders is assigned to a driver, and the driver pick-ups each rider by traveling to each rider's origin, then drop-of them at some public transit stops.
For a post-transit rideshare match, a driver picks-up a set of riders at a public transit stop and then transport each rider one by one to their destinations.
A set of riders can only be assigned to a driver if certain constraints are met, such as capacity of the driver's vehicle and the travel time constraints of the driver and riders.
Each of the driver and rider is represented by a node. There is an edge between a driver and a rider if the rider can be served by the driver.
If a group of riders can be served by a driver, a node containing the group is created, and an edge between the driver and the group is also created.
From this graph, a matching problem is formulated as an integer linear program (ILP) and solved by standard branch and bound (CPLEX).
The optimization goal in~\cite{EEEIC17-M} is to minimize cost related to waiting time and travel time, but ridesharing routes are not guarantee to be better than transit route.
Although the optimization goal in~\cite{COR18-S} aligns with ours, there are some limitations in their approach:
they limit at most two riders for each rideshare match, each rider must travel to the transit stop that is closest to the rider's destination, and more importantly, ridesharing routes assigned to riders can be longer than public transit routes.

In this paper, we use a similar model as in~\cite{EEEIC17-M,COR18-S}. We extend the work in~\cite{COR18-S} to eliminate the limitations described above and give approximation algorithms for the optimization problem to ensure solution quality.
Our discrete algorithms allow to control the trade-off between quality and computational time.
We conduct a numerical study based on real-life data in Chicago City.
Our main contributions are summarized as follows:
\begin{enumerate}
\setlength\itemsep{0em}
\item We give an exact algorithm approach (an ILP formulation based on a hypergraph representation) for integrating public transit and ridesharing.
\item We prove our maximization problem is NP-hard and give a 2-approximation algorithm for the problem. We show that previous $O(k)$-approximation algorithms~\cite{SWAT00-B,SODA99-C} for the $k$-set packing problem are 2-approximation algorithms for our maximization problem. Our algorithm is more time and space efficient than previous algorithms.
\item As a case study, we conduct an extensive numerical study based on real-life data in Chicago City to evaluate the potential of having an integrated transit system and the effectiveness of different approximation algorithms.
\end{enumerate}
The rest of the paper is organized as follows.
In Section~\ref{sec-preliminary}, we give the preliminaries of the paper, describe a centralized system that integrates public transit and ridesharing, and define the maximization problem.
In Section~\ref{sec-exact}, we describe our exact algorithm approach. We then propose approximation algorithms in Section~\ref{sec-approximate}.
We discuss our numerical experiments and results in Section~\ref{sec-experiment}.
Finally, Section~\ref{sec-conclusion} concludes the paper.

\section{Problem definition and preliminaries} \label{sec-preliminary}
In the problem \textit{multimodal transportation with ridesharing} (MTR), we have a centralized system, and for every fixed time interval, the system receives a set $\mathcal{A} = D \cup R$ of trips with $D \cap R = \emptyset$, where $D$ is the set of driver trips and $R$ is the set of rider trips.
Each trip is expressed by an integer label $i$ and consists of an individual, a vehicle (for driver trip) and some requirements.
A connected public transit network with a fixed timetable $T$ is given.
We assume that for any source $o$ and destination $d$ in the public transit network, $T$ gives the fastest travel time from $o$ to $d$.
A \emph{ridesharing route} $\pi_i$ for a rider $i \in R$ is a travel plan using a combination of public transportation and ridesharing to reach $i$'s destination satisfying $i$'s requirements, whereas a \emph{public transit route} $\hat{\pi}_i$ for a rider $i$ is a travel plan using only public transportation.
The multimodal transportation with ridesharing problem asks to provide at least one feasible route ($\pi_i$ or $\hat{\pi}_i$) for every rider $i \in R$. We denote an instance of multimodal transportation with ridesharing problem by $(N,\mathcal{A},T)$, where $N$ is an edge-weighted directed graph (network) for both private and public transportation.
We call a public transit station or stop just \emph{station}.
The terms rider and passenger are used interchangeably (although passenger emphasizes a rider has been provided with a ridesharing route).

The requirements of each trip $i$ in $\mathcal{A}$ are specified by $i$'s parameters submitted by the individual.
The parameters of a trip $i$ contain an origin location $o_i$, a destination location $d_i$, an earliest departure time $\alpha_i$, a latest arrival time $\beta_i$ and a maximum trip time $\gamma_i$.
A driver trip $i$ also contains a capacity $n_i$ of the vehicle, a limit $\delta_i$ on the number of stops a driver wants to make to pick-up/drop-off passengers, and an optional path to reach its destination.
The maximum trip time $\gamma_i$ of a driver $i$ includes a travel time from $o_i$ to $d_i$ and a detour time limit $i$ can spend for offering ridesharing service.
A rider trip $i$ also contains an acceptance rate $\theta_i$ for a ridesharing route $\pi_i$, that is, $\pi_i$ is given to rider $i$ if $t(\pi_i) \leq \theta_i \cdot t(\hat{\pi}_i)$ for every public transit route $\hat{\pi}_i$ and $0 < \theta_i \leq 1$, where $t(\cdot)$ is the travel time.
Such a route $\pi_i$ is called an \emph{acceptable ridesharing route} (acceptable route for brevity).
For example, suppose the best public transit route $\hat{\pi}_i$ takes 100 minutes for $i$ and $\theta_i = 0.9$. An acceptable route $\pi_i$ implies that $t(\pi_i) \leq \theta_i \cdot t(\hat{\pi}_i) = 90$ minutes.
We consider two match types for practical reasons.
\begin{itemize}
\setlength\itemsep{0em}
\item \textbf{Type 1 (rideshare-transit)}: a driver may make multiple stops to pick-up different passengers, but makes only one stop to drop-off all passengers. In this case, the \emph{pick-up locations} are the passengers' origin locations, and the \emph{drop-off location} is a public station.
\item \textbf{Type 2 (transit-rideshare)}: a driver makes only one stop to pick-up passengers and may make multiple stops to drop-off all passengers. In this case, the \emph{pick-up location} is a public station and the \emph{drop-off locations} are the passengers' destination locations.
\end{itemize}
Riders and drivers specify one of the match types to participate in; they are allowed to choose both in hope to increase the chance being selected, but the system will assign them only one of the match types such that the optimization goal of the MTR problem is achieved, which is to assign acceptable routes to as many riders as possible.
Formally, the \textbf{maximization problem} we consider is to maximize the number of passengers, each of whom is assigned an acceptable route $\pi_i$ for every $i \in R$.

For a driver $i$ and a set $J \subseteq R$ of riders, $\sigma(i) = \{i\} \cup J$ is called a {\em feasible match} if the routes for all trips of $\sigma(i)$ satisfy the requirements (constraints) specified by the parameters of the trips collectively as listed below (a summary of notation and constraints can be found in Section~\ref{sec-compute-matches}).
\begin{enumerate}
\setlength\itemsep{0em}
\item \textit{Ridesharing route constraint}: for $J=\{j_1,\ldots,j_k\}$, there is a path $(o_i,o_{j_1},\dots,o_{j_k}$, $s,d_i)$ in $N$, where $s$ is the drop-off location for Type 1 match;
or there is a path $(o_i,s,d_{j_1},...,d_{j_k},d_i)$ in $N$, where $s$ is the pick-up location for Type 2 match.
\item \textit{Capacity constraint}: $1\leq |J| \leq n_i$.
\item \textit{Acceptable constraint}: each passenger $j \in J$ is given an acceptable route $\pi_j$ offered by driver $i$.
\item \textit{Travel time constraint}: each trip $j \in \sigma(i)$ departs from $o_j$ no earlier than $\alpha_j$, arrives at $d_j$ no later than $\beta_j$, and the total travel duration of $j$ is at most $\gamma_j$. 
\item \textit{Stop constraint}: the number of unique locations visited by driver $i$ to pick-up (for Type 1) or drop-off (for Type 2) all passengers of $\sigma(i)$ is at most $\delta_i$.
\end{enumerate}
Two feasible matches $\sigma(i), \sigma(i')$ are \emph{disjoint} if $\sigma(i) \cap \sigma(i') = \emptyset$.
Then, the maximization problem considered is to find a set of pairwise disjoint feasible matches such that the number of passengers included in the feasible matches is maximized.

Intuitively, a rideshare-transit (Type 1) feasible match $\sigma(i)$ is that all passengers in $\sigma(i)$ are picked-up at their origins and dropped-off at a station, and then $i$ drives to destination $d_i$ while each passenger $j$ of $\sigma(i)$ takes transit to destination $d_j$.
A transit-rideshare (Type 2) feasible match $\sigma(i)$ is that all passengers in $\sigma(i)$ are picked-up at a station and dropped-off at their destinations, and then $i$ drives to destination $d_i$ after dropping the last passenger.
We give algorithms to find pairwise disjoint feasible matches to maximize the number of passengers included in the matches.
We describe our algorithms for Type 1 only. Algorithms for Type 2 can be described with the constraints on the drop-off location and pick-up location of a driver exchanged, and we omit the description.
Further, it is not difficult to extend to other match types, such as rideshare only and park-and-ride, as described in~\cite{COR18-S}.

\section{Exact algorithm} \label{sec-exact}
An exact algorithm for the maximization problem is presented in this section, which is similar to the matching approach described in~\cite{PNAS17-AM,PNAS14-S} for ridesharing and in~\cite{EEEIC17-M,COR18-S} for MTR.

\subsection{Integer program formulation} \label{sec-exact-IP}
The exact algorithm is summarized as follows.
First, we compute all feasible matches for each driver $i$. Then, we create a bipartite (hyper)graph $H(D,R,E)$, where $D(H)$ is the set of drivers, and $R(H)$ is the set of riders.
There is a hyperedge $e = (i, J)$ in $E(H)$ between $i \in D(H)$ and a non-empty subset $J \subseteq R(H)$ if $\{i\} \cup J$ is a feasible match, denoted by $\sigma_J(i)$, for driver $i$.
\begin{figure}[!htbp]
\centering
\includegraphics[width=.66\linewidth]{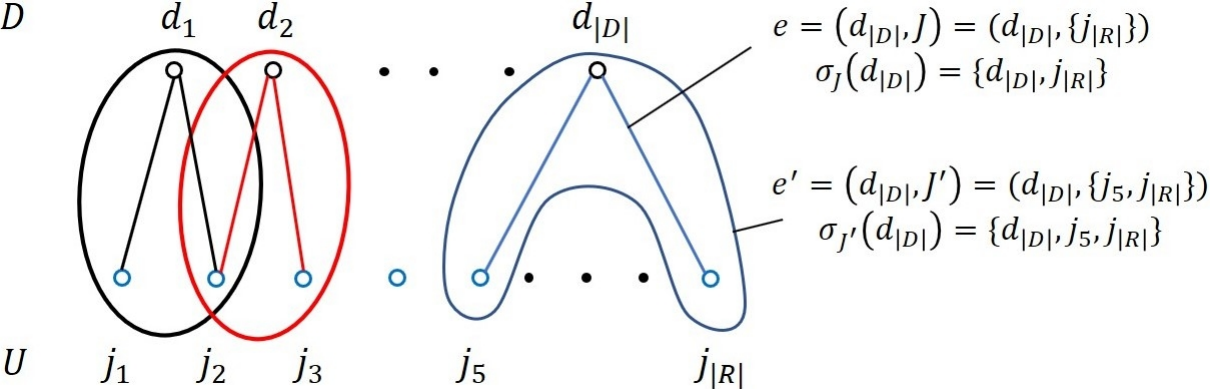}
\caption{A bipartite hypergraph for all possible matches of an instance $(N,\mathcal{A},T)$.}
\label{fig-hypergraph}
\end{figure}
An example is given in Figure~\ref{fig-hypergraph}.
Any driver $i$ and rider $j$ with no feasible match is removed from $D(H)$ and $R(H)$ respectively, namely, no isolated vertex.
For an edge $e=(i,J)$, let $A(e) = \{i\} \cup J$ and $p(e) = |J|$ be the number of riders represented by $e$.
For a trip $j \in \mathcal{A}$, define $E_j = \{e \in E \mid j \in A(e)\}$ to be the set of edges in $E$ associated with $j$.
To solve the maximization problem, we give an integer program (ILP) formulation:
\begin{alignat}{4}
 & \text{maximize }   &        &\sum_{e \in E(H)} p(e) \cdot x_{e} & \qquad \label{obj-1}\\
 & \text{subject to } & \qquad &\sum_{e \in E_j} x_{e} \leq 1,   & & \forall \text{ } j \in \mathcal{A} \label{constraint-1}\\
 &                    &        &x_{e} \in \{0,1\}, & &\forall \text{ } e \in E(H) \label{constraint-2}
\end{alignat}
The binary variable $x_e$ indicates whether the edge $e = (i, J)$ is in the solution ($x_e = 1$) or not ($x_e = 0$).
If $x_e = 1$, it means that all passengers in $J$ are served by $i$.
Inequality (2) in the ILP formulation guarantees that each driver serves at most one feasible set of passengers and each passenger is served by one driver.
Note that the ILP (\ref{obj-1})-(\ref{constraint-2}) is similar to a set packing formulation.
An advantage of this ILP formulation is that the number of constraints is substantially decreased, compared to traditional ridesharing formulation.

\begin{observation}
A match $\sigma(i)$ for any driver $i \in D$ is feasible if and only if for every subset $P$ of $\sigma(i)\setminus\{i\}$, the match between $i$ and $P$ is feasible~\cite{TRBM15-S}.
\label{obs-1}
\end{observation}

From Observation~\ref{obs-1}, it is not difficult to see that the following results hold.
\begin{prop}
Let $i_1, i_2,\ldots, i_j$ be a set of drivers in $D$ and $P$ be a maximal set of passengers served by $i_1, \ldots, i_j$.
There always exists a solution such that $\sigma(i_a) \cap \sigma(i_b) = \emptyset$ $(1 \leq a \neq b \leq j)$ and $\bigcup_{i_1 \leq a \leq i_j} \sigma(a) = P$.
\label{prop-1}
\end{prop}

\begin{theorem}\label{theorem-ILP}
Given a bipartite graph $H(D,R,E)$ representing an instance of the multimodal transportation with ridesharing maximization problem, 
an optimal solution to the ILP (\ref{obj-1})-(\ref{constraint-2}) is an optimal solution to the maximization problem and vice versa.
\end{theorem}

\begin{proof}
From inequality (\ref{constraint-1}) in the integer program, the solution found by the integer program is always feasible to the maximization problem.
By Proposition~\ref{prop-1} and objective function~(\ref{obj-1}), an optimal solution to the ILP (\ref{obj-1})-(\ref{constraint-2}) is an optimal solution to the maximization problem.
Obviously, an optimal solution to the maximization problem is an optimal solution to the ILP (\ref{obj-1})-(\ref{constraint-2}).
\end{proof}

\subsection{Computing feasible matches} \label{sec-compute-matches}
Let $i$ be a driver in $D$ and $n_i$ be the capacity of $i$ (maximum number of riders $i$ can serve). The maximum number of feasible matches for $i$ is $\sum_{p = 1}^{n_i} \binom{|R|}{p}$.
Assuming the capacity $n_i$ is a very small constant (which is reasonable in practice), the above summation is polynomial in $R$, that is, $O((|R|+1)^{n_i})$ (partial sums of binomial coefficients).
Let $K = \max_{i \in D} {n_i}$ be the maximum capacity among all vehicles (driver trips).
Then, in the worst case, $|E(H)| = O(|D| \cdot (|R|+1)^K)$.

We compute all feasible matches for each trip in two phases. In phase one, for each driver $i$, we find all feasible matches $\sigma(i)=\{i,j\}$ with one rider $j$.
In phase two, for each driver $i$, we compute all feasible matches $\sigma(i)=\{i,j_1,..,j_p\}$ with $p$ riders, based on the feasible matches $\sigma(i)$ with $p-1$ riders computed previously, for $p=2$ and upto the number of passengers $i$ can serve.
Before describing how to compute the feasible matches, we first introduce some notations and specify the feasible match constraints we consider.
Each trip $i \in \mathcal{A}$ is specified by the parameters $(o_i, d_i, n_i, z_i, p_i, \delta_i, \alpha_i, \beta_i, \gamma_i, \theta_i)$, where the parameters are summarized in Table~\ref{table-notation} along with other notation.
\begin{table}[!ht]
\footnotesize
\centering
   \begin{tabular}{| c | l |}
   	\hline
   	\textbf{Notation}  & \textbf{Definition}                                                    \\ 
   	$o_i$              & Origin (start location) of $i$ (a vertex in $N$)                  		\\
   	$d_i$              & Destination of $i$ (a vertex in $N$)                              		\\
   	$n_i$              & Number of seats (capacity) of $i$ available for passengers (driver only)   \\
   	$z_i$              & Maximum detour time $i$ willing to spend for offering ridesharing services (driver only)    \\
   	$p_i$       	       & An optional preferred path of $i$ from $o_i$ to $d_i$ in $N$ (driver only)     	\\
   	$\delta_i$         & Maximum number of stops $i$ willing to make to pick-up passengers for match     \\
   	                       & Type 1 and to drop-off passengers for match Type 2. \\
   	$\alpha_i$         & Earliest departure time of $i$                                     	\\
   	$\beta_i$          & Latest arrival time of $i$                                         	\\
    $\gamma_i$         & Maximum trip time for $i$                           					\\
    $\theta_i$         & Acceptance rate ($0 \leq \theta_i < 1$) for a ridesharing route $\pi_i$ (rider only)  \\
    $\pi_i$      	   & Route for $i$ using a combination of public transit and ridesharing (rider only) \\
    $\hat{\pi}_i$      & Route for $i$ using only public transit (rider only) 					\\
    $d(\pi_i)$         & The driver of ridesharing route $\pi_i$		           				\\
    $t(p_i)$           & Travel time for traversing path $p_i$ by private vehicle  				\\
    $t(\pi_i)$ \& $t(\hat{\pi}_i)$   & Travel time for traversing route $\pi_i$ and $\hat{\pi}_i$ resp.   \\
    $t(u,v)$ \& $\hat{t}(u,v)$   & Travel time from location $u$ to $v$ by private vehicle and public transit resp. \\ \hline
   \end{tabular}
\caption{Parameters for a trip announcement $i$.}
\label{table-notation}
\end{table}
The maximum trip time $\gamma_i$ of a driver $i$ can be calculated as $\gamma_i = t(p_i) + z_i$ if $p_i$ is given; otherwise $\gamma_i = t(o_i,d_i) + z_i$. For a passenger $j$, $\gamma_j$ is more flexible; it is default to be $\gamma_j = t(\hat{\pi}_i)$, where $\hat{\pi}_i$ is the fastest public transit route.

For a driver $i \in D$ and a set $J \subseteq R$ of riders, the set $\sigma(i) = \{i\} \cup J$ is called a {\em feasible match} if the routes for all trips of $\sigma(i)$ satisfy the requirements (constraints) specified by the parameters of the trips collectively as listed below:
\begin{enumerate}
\setlength\itemsep{0em}
\item \textit{Ridesharing route constraint}: for $J=\{j_1,\ldots,j_k\}$, there is a path $(o_i,o_{j_1},\ldots,o_{j_k}$, $s,d_i)$ in $N$, where $s$ is the drop-off location for Type 1 match;
or there is a path $(o_i,s,d_{j_1},...,d_{j_k},d_i)$ in $N$, where $s$ is the pick-up location for Type 2 match.
\item \textit{Capacity constraint}: limits the number of passengers a driver can serve, $1\leq |J| \leq n_i$.
\item \textit{Acceptable constraint}: each passenger $j \in J$ is given an acceptable route $\pi_j$ offered by $i$ such that $t(\pi_j) \leq \theta_j \cdot t(\hat{\pi}_j)$ for $0 < \theta_j \leq 1$, where $\hat{\pi}_j$ is the public transit route with shortest travel time for $j$.
\item \textit{Travel time constraint}: each trip $j \in \sigma(i)$ departs from $o_j$ no earlier than $\alpha_j$, arrives at $d_j$ no later than $\beta_j$, and the total travel duration of $j$ is at most $\gamma_j$.
The exact application of these time constraints is described in Subsection~\ref{subsection-alg-feas-all} (Algorithm 1) and Subsection~\ref{subsection-alg-feas-all} (Algorithm 2).
\item \textit{Stop constraint}: the number of unique locations visited by driver $i$ to pick-up (for Type 1) or drop-off (for Type 2) all passengers of $\sigma(i)$ is at most $\delta_i$.
\end{enumerate}


\subsubsection{Phase one (Algorithm 1)}\label{subsection-alg-feas-single}
Now we describe how to compute a feasible match between a driver and a passenger for Type 1. The computation for Type 2 is similar and we omit it.
For every trip $i \in D \cup R$, we first compute the set $S_{do}(i)$ of feasible drop-off locations for trip $i$.
Each element in $S_{do}(i)$ is a station-time tuple $(s, \alpha_i(s))$ of $i$, where $\alpha_i(s)$ is the earliest possible time $i$ can reach station $s$.
When computing feasible matches, we use a simplified model for the waiting time and ridesharing service time: given the fastest travel time $t(u,v)$ from location $u$ to location $v$, we multiply a small constant $\epsilon>1$ with $t(u,v)$ to simulate the waiting time and ridesharing service time.
In this model, the waiting time and ridesharing service time are considered together, as a whole.
The station-time tuples are computed by the following preprocessing procedure.
\begin{itemize}[leftmargin=*]
\setlength\itemsep{0em}
\item We find all feasible station-time tuples for each passenger $j \in R$. A station $s$ is \emph{feasible} for $j$ if $j$ can reach $d_j$ from $s$ within time window $[\alpha_j, \beta_j]$, $t(o_j,s) + \hat{t}(s,d_j) \leq \gamma_j$ and $t(o_j,s) + \hat{t}(s,d_j) \leq \theta_j \cdot \hat{t}(o_j,d_j)$.
	\begin{itemize}
	\item The earliest possible time to reach station $s$ for $j$ can be computed as $\alpha_j(s) = \alpha_j + t(o_j,s)$ without pick-up and drop-off time. Since we do not consider waiting time and ridesharing service time separately, $\alpha_j(s)$ also denotes the earliest departure time of $j$ at station $s$.
	\item Let $\hat{t}(s,d_j)$ be the travel time of a fastest public route. Station $s$ is \emph{time feasible} if $\alpha_j(s) + \hat{t}(s,d_j) \leq \beta_j$, $t(o_j,s) + \hat{t}(s,d_j) \leq \gamma_j$ and $t(o_j,s) + \hat{t}(s,d_j) \leq \theta_j \cdot \hat{t}(o_j,d_j)$.
	\end{itemize}
\item Next, we find all feasible station-time tuples for each driver $i \in D$ using a similar calculation.
	\begin{itemize}[leftmargin=*]
    \setlength\itemsep{0em}
    \item Without considering pick-up and drop-off time separately, the earliest arrival time of $i$ to reach $s$ is $\alpha_i(s) = \alpha_i + t(o_i,s)$. Station $s$ is \emph{time feasible} if $\alpha_i(s) + t(s,d_i) \leq \beta_i$ and $t(o_i,s) + t(s,d_i) \leq \gamma_i$.
	\end{itemize}
\end{itemize}

After the preprocessing, the Algorithm~1 finds all matches consists of a single passenger.
For each pair $(i, j)$ in $D \times R$, let $\alpha_i(o_j) = \max\{\alpha_i,\alpha_j - t(o_i,o_j)\}$ be the latest departure time for driver $i$ from $o_i$ such that $i$ can still pick-up $j$ at the earliest; this minimizes the time (duration) needed for driver $i$ to wait for passenger $j$, and hence, the total travel time of $i$ is minimized.
The process of checking if the match $\sigma(i) = \{i,j\}$ is feasible for all pairs of $(i,j)$ can be performed as in Algorithm~1 in Figure~\ref{alg-feas-single}.
\begin{figure}[htbp]
\small
\textbf{Algorithm~1} Single passenger
\begin{algorithmic}[1]
\For {each pair $(i, j)$ in $D \times R$}
    \For {each station $s$ in $S_{do}(i) \cap S_{do}(j)$}
        \State $t_1 = t(o_i,o_j) + t(o_j,s)$; $t_2 = t(o_j,s)$; \hspace*{2mm} /* travel duration for $i$ and $j$ to reach $s$ resp. */
        \State $t = \alpha_i(o_j) + t_1$; \hspace*{6mm} /* earliest departure time at station $s$ */
        \If {$t \leq \beta_i(s) \wedge (t_1 + t(s, d_i) \leq \gamma_i) \wedge t \leq \beta_j(s) \wedge (t_2 + \hat{t}(s, d_j) \leq \gamma_j) \wedge  (t_2 + \hat{t}(s, d_j) \leq \theta_j \cdot \hat{t}(o_j, d_j)\})$}
            \State create an edge $(i, J=\{j\})$ in $E(H)$ to represent $\sigma(i) = \{i,j\}$.
            \State \textbf{break} inner for-loop; \hspace*{2mm} /* can be allowed to run to completion for a better route */
        \EndIf
    \EndFor
\EndFor
\end{algorithmic}
\caption{Algorithm for computing matches consists of a single passenger.}
\label{alg-feas-single}
\vspace{-1mm}
\end{figure}

\subsubsection{Phase two (Algorithm 2)}\label{subsection-alg-feas-all}
We extend Algorithm~1 to create matches with more than one passenger.
Let $H(D,R,E)$ be the graph after computing all possible matches consists of a single passenger (instance computed by Algorithm~1).
We start with computing feasible matches consists of two passengers, then three passengers, and so on.
Let $\varSigma(i)$ be the set of matches found so far for driver $i$ and $\varSigma(i,p-1) = \{\sigma(i) \in \varSigma(i) \mid |\sigma(i) \setminus \{i\}|=p-1\}$ be the set of matches with $p-1$ passengers, and we try to extend $\varSigma(i,p-1)$ to $\varSigma(i,p)$ for $p \geq 2$.
Let $r_i = (l_0,l_1,\ldots,l_p,s)$ denotes an ordered potential path (travel route) for driver $i$ to pick-up all $p$ passengers of $\sigma(i)$ and drop-off them at station $s$, where $l_0$ is the origin of $i$ and $l_y$ is the pick-up location (origin of passenger $j_y$), $1 \leq y \leq p$.
We extend the notion of $\alpha_i(o_j)$, defined above, to all locations of $r_i$.
That is, $\alpha_i(l_p)$ is the latest departure time of $i$ to pick-up all passengers $j_1,\ldots,j_p$ such that the waiting time of $i$ is minimized, and hence, travel time of $i$ is minimized.
All possible combinations of $r_i$ are enumerated to find a feasible path $r_i$; the process of finding $r_i$ is described in the following.
\begin{itemize}
\setlength\itemsep{0em}
\item First, we fix a combination of $r_i$ such that $|\sigma(i)| \leq n_i + 1$ and $r_i$ satisfies the stop constraint. The order of the pick-up origin locations is known when we fix a path $r_i$.
\item The algorithm determines the actual drop-off station $s$ in $r_i = (l_0,l_1,\ldots,l_{p},s)$.
Let $j_{y}$ be the passenger corresponds to pick-up location $l_y$ for $1 \leq y \leq p$ and $l_0 = o_i$.
For each station $s$ in $\bigcap_{0 \leq y \leq p} S_{do}(j_y)$, the algorithm checks if $r_i = (l_0,l_1,\ldots,l_{p},s)$ admits a time feasible path for each trip in $\sigma(i)$.
    \begin{itemize}
        \item The total travel time (duration) for $i$ from $l_0$ to $s$ is $t_i = t(l_0, l_1) + \cdots + t(l_{p-1},l_{p}) + t(l_p, s)$.
        The total travel time (duration) for $j_y$ from $l_y$ to $s$ is $t_{j_y} = t(l_y,l_{y+1}) + \cdots + t(l_{p-1},l_{p}) + t(l_p, s)$, $1 \leq y \leq p$.
        \item Since the order for $i$ to pick up $j_y$ ($1 \leq y \leq p$) is fixed, $\alpha_i(l_p)$ can be calculated as $\alpha_i(l_p) = \max\{\alpha_i, \alpha_{j_1} - t(l_0,l_1), \alpha_{j_{2}} - t(l_0,l_1) - t(l_1,l_2), \ldots, \alpha_{j_{p}} - t(l_0,l_1) - \cdots - t(l_{p-1},l_p)\}$.
        The earliest arrival time at $s$ for all trips in $\sigma(i)$ is $t = \alpha_i(l_p) + t_i$.
        \item If $t \leq \beta_i(s)$, $t_i + t(s, d_i) \leq \gamma_i$, and for $1\leq y\leq p$, $t \leq \beta_{j_{y}}(s)$, $t_{j_y} + \hat{t}(s, d_{j_{y}}) \leq \gamma_{j_{y}}$ and $t_{j_y} + \hat{t}(s, d_{j_{y}}) \leq \delta_{j_{y}} \cdot \hat{t}(o_{j_{y}}, d_{j_{y}})$, then $r_i$ is feasible.
    \end{itemize}
\item If $r_i$ is feasible, add the match corresponds to $r_i$ to $H$. Otherwise, check next combination of $r_i$ until a feasible path $r_i$ is found or all combinations are exhausted.
\end{itemize}
The pseudo code for the above process is given in Algorithm~2 (Figure~\ref{alg-feas-all}).
\begin{figure}[!htbp]
\small
\textbf{Algorithm~2} Compute all matches
\begin{algorithmic}[1]
\For {$i$ = 1 to $|D|$}
    \State $p = 2$;
    \While {($p \leq n_i$ and $\varSigma(i,p-1) \neq \emptyset$)}
    \For {each match $\sigma(i)$ in $\varSigma(i,p-1)$}
        \For {each $j \in R$ s.t. $j \notin \sigma(i)$}
            \State /* check if $\sigma(i) \cup \{j\}$ satisfies Observation~\ref{obs-1}, and if not, skip $j$ */
            \State \textbf{if} {$((\sigma(i) \setminus \{q\}) \cup \{j\}) \in \varSigma(i,p-1)$ for all $q \in \sigma(i) \setminus \{i\}$} \textbf{then}
                \State \hspace{0.5cm} \textbf{if} {($\sigma(i) \cup \{j\}$ has not been checked) and (feasibleInsert($\sigma(i), j$))} \textbf{then}
                    \State \hspace{1cm} create an edge $(i, J)$ in $E(H)$ to represent $\sigma_J(i) = \{i,J\}$.
                    \State \hspace{1cm} add $\sigma(i) \cup \{j\}$ to $\varSigma(i,p)$.
                \State \hspace{0.5cm} \textbf{end if}
         \EndFor
    \EndFor
    \State $p = p + 1$;
    \EndWhile
\EndFor
\\
\textbf{Procedure} feasibleInsert($\sigma(i), j$) \hspace{3mm} /* find a feasible path for $i$ to serve $\sigma(i) \cup \{j\}$ if exists */
\\ Let $r_i = (l_0,l_1,\ldots,l_p,s)$ denotes a potential path for driver $i$ to serve trips in $\sigma(i) \cup \{j\}$.
\For {each station $s$ in $\bigcap_{0 \leq y \leq p} S_{do}(j_y)$}
    \For {each combination of $r_i = (l_0,\ldots,l_p,s)$ that satisfies the stop constraint}
        \State $t_i = t(l_0, l_1) + \cdots + t(l_{p-1},l_{p}) + t(l_p, s)$; $t_{j_y} = t(l_y,l_{y+1}) + \cdots + t(l_{p-1},l_{p}) + t(l_p, s)$;
        \State $t = \alpha_i(l_p) + t_i$; /* the earliest arrival time at $s$ for all trips in $\sigma(i)$ */
        \State \textbf{if} [($t \leq \beta_i(s) \wedge (t_i + t(s, d_i) \leq \gamma_i)$) and (for $1\leq y \leq p$, $t \leq \beta_{j_{y}}(s) \wedge (t_{j_y} + \hat{t}(s, d_{j_{y}}) \leq$
        \State \hspace{0.5cm} $\theta_{j_{y}} \cdot \hat{t}(o_{j_{y}}, d_{j_{y}})) \wedge (t_{j_y} + \hat{t}(s, d_{j_{y}}) \leq \gamma_{j_{y}}))$] \textbf{then}
            \State \hspace{0.5cm} \Return True;
    \EndFor
\EndFor
\\ \Return False;
\end{algorithmic}
\caption{Algorithm for computing matches consists of multiple passengers.}
\label{alg-feas-all}
\vspace{-1mm}
\end{figure}
We show that the latest departure time $\alpha_i(l_p)$ used in Algorithm~2 indeed minimizes the total travel time of $i$ to reach $l_p$.
\begin{theorem}
Given a feasible path $r_i = (l_0,\ldots,l_p,s)$ for driver $i$ that serves $p$ passengers in a match $\sigma(i)$.
The latest departure time $\alpha_i(l_p)$ calculated above minimizes the total travel time of $i$ to reach $l_p$.
\end{theorem}

\begin{proof}
Prove by induction. For the base case $\alpha_i(l_1) = \max\{\alpha_i,\alpha_{j_{1}} - t(l_0,l_{1})\}$, $i$ does not need to wait for $j_1$. Hence, the total travel time of $i$ to pick-up $j_1$ is minimized with departure time $\alpha_i(l_1)$.
Assume the lemma holds for $1 \leq y-1 < p$, that is, $\alpha_i(l_{y-1})$ minimizes the total travel time of $i$ to reach $l_{y-1}$. We prove for $y$.
From the calculation, $\alpha_i(l_{y}) = \max\{\alpha_i(l_{y-1}), \alpha_{j_{y}} - t(l_0,l_{1}) - t(l_{1},l_{2}) - \cdots - t(l_{y-1},l_y)\}$. By the induction hypothesis, $\alpha_i(l_{y})$ minimizes the total travel time of $i$.
\end{proof}

The running time of Algorithm~2 heavily depends on the number of subsets of passengers to be checked for feasibility.
One way to speed up Algorithm~2 is to use dynamic programming (or memoization) to avoid redundant checks on a same subset.
For each feasible match $|\sigma(i)| = p$ of a driver $i \in D$, we store every feasible path $r_i = \{i, j_1,\ldots,j_p,s\}$ and extend from each feasible path $r_i$ to insert a new trip to minimize the number of ordered potential paths we need to test.
We can further make sure that no path is tested twice during execution.
First, the set $R$ of riders is given a fixed ordering (based on the integer labels).
For a feasible path $r_i$ of a driver $i$, the check of
inserting a new rider $j$ into $r_i$ is performed only if $j$ is larger than every rider in $r_i$ according to the fixed ordering.
A heuristic approach to speed up Algorithm~2 is given in at the end of Section~\ref{sec-instances}.

\section{Approximation algorithms} \label{sec-approximate}
We show that the maximization problem defined in Section 3 is NP-hard and give approximation algorithms for the problem. 
When every edge in $H(D,R,E)$ consists of only two vertices (one driver and one passenger), the maximization problem is equivalent to the maximum matching, which can be solved in polynomial time.
However, if the edges consist of more than two vertices, they become hyperedges. In this case, the integer program~(\ref{obj-1})-(\ref{constraint-2}) becomes a formulation of the maximum weighted set packing problem, which is NP-hard~~\cite{CINP79-GJ,Karp72}.
Our maximization problem is a special case of the maximum weighted set packing problem.
We first show our maximization problem instance $H(D,R,E)$ is indeed NP-hard.

\subsection{NP-hardness}
It was mentioned in~\cite{PNAS14-S} that their minimization problem related to shareability hyper-network is NP-Complete, which is similar to our maximization problem formulation. However, an actual reduction proof was not described.
One may think the maximization problem relates to covering problems; rather, it relates to packing problems.
We prove our maximization problem is NP-hard by a reduction from a special case of the maximum 3-dimensional matching problem (3DM).
An instance of 3DM consists of three disjoint finite sets $A$, $B$ and $C$, and a collection $\mathcal{F} \subseteq A \times B \times C$.
That is, $\mathcal{F}$ is a collection of triplets $(a,b,c)$, where $a \in A, b \in B$ and $c \in C$.
A 3-dimensional matching is a subset $\mathcal{M} \subseteq \mathcal{F}$ such that all sets in $\mathcal{M}$ are pairwise disjoint.
The decision problem of 3DM is that given $(A, B, C, \mathcal{F})$ and an integer $q$, decide whether there exists a matching $\mathcal{M} \subseteq \mathcal{F}$ with $|\mathcal{M}| \geq q$.
We consider a special case of 3DM: $|A| = |B| = |C| = q$; it is still NP-complete~\cite{CINP79-GJ,Karp72}.
Given an instance $(A,B,C,\mathcal{F})$ of 3DM with $|A| = |B| = |C| = q$, we construct an instance $H(D,R,E)$ (bipartite hypergraph) of the maximization problem as follows:
\begin{itemize}
\setlength\itemsep{0em}
\item $D(H) = A$, the set of drivers and $R(H) = B \cup C$, the set of passengers.
\item For each $f \in \mathcal{F}$, create a hyperedge $e(f)$ in $E(H)$ containing elements $(a,b,c)$, where $a$ represents a driver and $\{b,c\}$ represent two different passengers.
Further, create edges $e'(f) = \{a, b\}$ and $e''(f) = \{a, c\}$.
\end{itemize}

\begin{theorem}
The maximization problem is NP-hard.
\end{theorem}

\begin{proof}
By Theorem~\ref{theorem-ILP}, we only need to prove the ILP~(\ref{obj-1})-(\ref{constraint-2}) is NP-hard, which is done by showing that an instance $(A,B,C,\mathcal{F})$ of the maximum 3-dimensional matching problem has a solution $\mathcal{M}$ of cardinality $q$ if and only if the bipartite hypergraph instance $H(D,R,E)$ has a solution $X$ with $2q$ passengers.

Assume that $(A,B,C,\mathcal{F})$ has a solution $\mathcal{M} = \{m_1, m_2,\ldots, m_q\}$. For each $m_i$ ($1 \leq i \leq q$), add the corresponding hyperedge $e(m_i) \in E(H)$ to $X$ (that is, setting the corresponding variable $x_{e(m_i)} = 1$).
Since $m_i \cap m_j = \emptyset$ for $1 \leq i \neq j \leq q$ (implying the constraint~(\ref{constraint-1}) of the ILP is always satisfied), and each edge $e \in X$ contains two passengers, $X$ is a valid solution to $H(D,R,E)$ with $2q$ passengers.

Assume that $H(D,R,E)$ has a solution $X$ with $2q$ passengers served (the objective function (\ref{obj-1}) of ILP is $2q$ and every edge $e(f) \in X$ corresponds to variable $x_{e(f)} = 1$).
For every edge $e(f) \in X$, add the corresponding set $f \in \mathcal{F}$ to $\mathcal{M}$.
From the constraint~(\ref{constraint-1}) of the ILP, $X$ is pairwise disjoint.
In order to serve $2q$ passengers, $|X| = |D| = q$ since every $e(f) \in X$ must contain two different passengers.
Hence, $\mathcal{M}$ is a valid solution to $(A,B,C,\mathcal{F})$ s.t. $|\mathcal{M}| = q$.

The size of $H(D,R,E)$ is polynomial in $q$. It takes a polynomial time to convert a solution of $H(D,R,E)$ to a solution of the 3DM instance $(A,B,C,\mathcal{F})$ and vice versa. 
\end{proof}

\subsection{2-approximation algorithm}
For consistency, we follow the convention in~\cite{SWAT00-B,SODA99-C} that a $\rho$-approximation algorithm for a maximization problem is defined as $\rho \cdot w(\mathcal{C}) \geq OPT$ for $\rho > 1$, where $w(\mathcal{C})$ and $OPT$ are the values of approximation and optimal solutions respectively.
In this section, we give a $2$-approximation algorithm to the maximization problem instance $H(D,R,E)$.
Our $2$-approximation algorithm (refer to as \textit{ImpGreedy}) is a simplified version of the simple greedy~\cite{SWAT00-B,SODA99-C,PNAS14-S} discussed in Section~\ref{sec-app-algs}, except the running time and memory usage are significantly improved by computing a solution directly from $H(D,R,E)$ without solving the independent set/weighted set packing problem.

\subsubsection{Description of ImpGreedy Algorithm}
For a maximization problem instance $H(D,R,E)$, we use $\Gamma$ to denote a current partial solution, which consists of a set of matches represented by the hyperedges in $E(H)$.
Let $P(\Gamma)=\bigcup_{e \in \Gamma} J_e$ (called \textit{covered passengers}).
Initially, $\Gamma = \emptyset$.
In each iteration, we add a match with the most number of uncovered passengers to $\Gamma$, that is, select an edge $e=(i,J_e)$ such that 
$|J_e \setminus P(\Gamma)|$ is maximum, and then add $e$ to $\Gamma$.
Remove $E_e = \cup_{j \in A(e)} E_j$ from $E(H)$ ($E_j$ is defined in Section~\ref{sec-exact-IP}).
Repeat until $P(\Gamma) = R$ or $|\Gamma| = |D|$.
The pseudo code of ImpGreedy algorithm is shown in Figure~\ref{alg-new-approx}.
\begin{figure}[htbp]
\small
\textbf{Algorithm~3} ImpGreedy Algorithm \\
\textbf{Input:} The hypergraph $H(D,R,E)$ for problem instance $(N,\mathcal{A},T)$. \\
\textbf{Output:} A solution $\Gamma$ to $(N,\mathcal{A},T)$ with $2$-approximation ratio.
\begin{algorithmic}[1]
\\ $\Gamma = \emptyset$; $P(\Gamma) = \emptyset$;
\While{($P(\Gamma) \neq R$ and $|V(\Gamma) < |D|$)}
    \State compute $e = \text{argmax}_{e \in E(H)} |J_e \setminus P(\Gamma)|$;
    $\Gamma = \Gamma \cup \{e\}$; update $P(\Gamma)$; remove $E_e$ from $E(H)$;
\EndWhile
\end{algorithmic}
\caption{$2$-approximating algorithm for problem instance $(N,\mathcal{A},T)$.}
\label{alg-new-approx}
\vspace{-1mm}
\end{figure}


\noindent In ImpGreedy Algorithm, when an edge $e$ is added to $\Gamma$, $E_e$ is removed from $E(H)$, so Property~\ref{property-gamma} holds for $\Gamma$.
\begin{property}
For every $i \in D$, at most one edge $e$ from $E_i$ can be selected in any solution.
\label{property-gamma}
\end{property}

\subsubsection{Analysis of ImpGreedy Algorithm}
Let $\Gamma = \{x_1, x_2,\ldots, x_a\}$ be a solution found by Algorithm~3, where $x_i$ is the $i^{th}$ edge added to $\Gamma$.
Throughout the analysis, we use $OPT$ to denote an optimal solution, that is, $P(OPT) \geq P(\Gamma)$.
Further, $\Gamma_i = \bigcup_{1 \leq b \leq i} x_b$ for $1 \leq i \leq a$, $\Gamma_0 = \emptyset$ and $\Gamma_a = \Gamma$.
The driver of match $x_i$ is denoted by $d(x_i)$.
The main idea of our analysis is to add up the maximum difference between the number of covered passengers by selecting $x_i$ in $\Gamma$ and not selecting $x_i$ in $OPT$.
For each $x_i\in \Gamma$, by Property~\ref{property-gamma}, there is at most one $y \in OPT$ with $d(y)=d(x_i)$.
We order $OPT$ and introduce dummy edges to $OPT$ such that $d(y_i) = d(x_i)$ for $1 \leq i\leq a$.
Formally, for $1\leq i\leq a$, define
\[
OPT(i)=\{y_1,\ldots,y_i \mid 1\leq b \leq i, d(y_b)=d(x_b) \text{ if } y_b \in OPT, \text{ otherwise } y_b \text{ a dummy edge}\}.
\] 
A dummy edge $y_b\in OPT(i)$ is defined as $d(y_b) = d(x_b)$ with $J_{y_b}=\emptyset$.
The gap of an edge $x_i \in \Gamma$ is defined as
\[
\gap(x_i) = |J_{y_i}| - |J_{x_i} \setminus P(\Gamma_{i-1})| + |J'_{x_i}|,
\]
where $J'_{x_i} = (J_{x_i} \setminus P(\Gamma_{i-1})) \cap P(OPT \setminus \Gamma)$ is the maximum subset of passengers in $J_{x_i} \setminus P(\Gamma_{i-1})$ that are also covered by drivers in $OPT \setminus \Gamma$.
The intuition is that the sum of $\gap(x_i)$ for all $x_i \in \Gamma$ states the maximum possible number of passengers may not be covered by $\Gamma$.
Let $P(OPT(i)) = \bigcup_{1 \leq b \leq i} J_{y_b}$ and $P(OPT'(i)) = \bigcup_{1 \leq b \leq i} J'_{x_b}$ for any $i \in [1,\ldots,a]$.
Then the maximum gap between $\Gamma$ and $OPT$ can be calculated as $\sum_{x \in \Gamma_a} \gap(x) = |P(OPT(a))| + |P(OPT'(a))| - |P(\Gamma_{a})|$.
First, we show that $P(OPT) = P(OPT(a)) \cup P(OPT'(a))$.

\begin{prop}
Let $\Gamma = \{x_1,\ldots,x_a\}$, $P(OPT(a)) = \bigcup_{1 \leq i \leq a} J_{y_i}$ and $P(OPT'(a)) = \bigcup_{1 \leq i \leq a} J'_{x_i}$.
Then, $P(OPT) = P(OPT(a)) \cup P(OPT'(a))$.
\label{prop-opt-size}
\end{prop}

\begin{proof}
By definition, $P(OPT)=P(OPT(a)) \cup P(OPT \setminus OPT(a))$.
For any $z$ in $OPT\setminus OPT(a)$, $d(z) \neq d(x)$ for every $x \in \Gamma$.
If $J_z \setminus P(\Gamma) \neq \emptyset$, then $z$ would have been found and added to $\Gamma$ by Algorithm~3.
Hence, $J_z \setminus P(\Gamma) = \emptyset$, implying $J_z \subseteq P(OPT'(a))$ and $P(OPT \setminus OPT(a)) \subseteq P(OPT'(a))$.
\end{proof}

\begin{lemma}
Let $OPT$ be an optimal solution and $\Gamma = \{x_1, x_2,\ldots, x_a\}$ be a solution found by the algorithm.
For any $1 \leq i \leq a$, $\sum_{x \in \Gamma_i} \gap(x) = |P(OPT(i))| - |P(\Gamma_{i})| + |P(OPT'(i))| \leq |P(\Gamma_i)|$.
\label{lemma-max-gap}
\end{lemma}

\begin{proof}
Recall that $OPT(i)=\{y_1,\ldots,y_i\}$ as defined above. For $y_b \in OPT(i), 1 \leq b \leq i, d(y_b)=d(x_b)$.
We prove the lemma by induction on $i$.
Base case $i=1$: $|P(OPT(1))| - |P(\Gamma_1)| + |P(OPT'(1))| \leq |P(\Gamma_1)|$.
By definition, $\gap(x_1) = |J_{y_1}| - |J_{x_1} \setminus \Gamma_0| + |J'_{x_1}|$.
Since $x_1$ is selected by the algorithm, it must be that $|J_{x_1}| \geq |J_u|$ for all $u \in V(G')$, so $|J_{y_1}| \leq |J_{x_1}|$.
Thus,
\begin{align*}
\gap(x_1) &= |J_{y_1}| - |J_{x_1} \setminus \Gamma_0| + |J'_{x_1}| \\
          &\leq |J'_{x_1}| \leq |J_{x_1}|.
\end{align*}
Assume the statement is true for $i-1 \geq 1$, that is, $\sum_{x \in \Gamma_{i-1}} \gap(x) \leq |P(\Gamma_{i-1})|$, and we prove for $i \leq a$.
By the induction hypothesis, both $P(OPT(i-1))$ and $P(OPT'(i-1))$ are included in the calculation of $\sum_{x \in \Gamma_{i-1}} \gap(x)$. More precisely, $\sum_{x \in \Gamma_{i-1}} \gap(x) = |P(OPT(i-1))| - |P(\Gamma_{i-1})| + |P(OPT'(i-1))| \leq |P(\Gamma_{i-1})|$.
If $|J_{y_i}| \leq |J_{x_i} \setminus P(\Gamma_{i-1})|$, the lemma is true since we can assume $|J'_{x_i}| \leq |J_{x_i}|$.
Suppose $|J_{y_i}| > |J_{x_i} \setminus P(\Gamma_{i-1})|$.
Before $x_i$ is selected, the algorithm must have considered $y_i$ and found that $|J_{x_i} \setminus P(\Gamma_{i-1})| \geq |J_{y_i} \setminus P(\Gamma_{i-1})|$.
Then, $|J_{y_i}| > |J_{x_i} \setminus P(\Gamma_{i-1})| \geq |J_{y_i} \setminus P(\Gamma_{i-1})|$, implying $J_{y_i} \cap P(\Gamma_{i-1}) \neq \emptyset$.
We have
\begin{align}
|J_{x_i} \setminus P(\Gamma_{i-1})| + |J_{y_i} \cap P(\Gamma_{i-1})| 
\geq |J_{y_i} \setminus P(\Gamma_{i-1})| + |J_{y_i} \cap P(\Gamma_{i-1})| = |J_{y_i}|.
\label{eq-4}
\end{align}

Let $J''_{y_i} \subseteq (J_{y_i} \cap P(\Gamma_{i-1}))$ be the set of passengers covered by $P(OPT(i-1)) \cup P(OPT'(i-1))$, namely $J''_{y_i} \subseteq (P(OPT(i-1)) \cup P(OPT'(i-1)))$.
Then by the induction hypothesis,
\begin{align}
\sum_{x \in \Gamma_{i-1}} \gap(x) \leq P(\Gamma_{i-1}) - |J_{y_i} \cap P(\Gamma_{i-1})| + |J''_{y_i}|.
\label{eq-5}
\end{align}
Adding $\sum_{x \in \Gamma_{i-1}} \gap(x)$ and $\gap(x_i)$ together:
\begin{align*}
&(\sum_{x \in \Gamma_{i-1}} \gap(x)) + (\gap(x_i)) \\
&= |P(OPT(i-1))| - |P(\Gamma_{i-1})| + |P(OPT'(i-1))| + |J_{y_i} \setminus J''_{y_i}| - |J_{x_i} \setminus P(\Gamma_{i-1})| + |J'_{x_i}| \\
&\leq (|P(\Gamma_{i-1})| - |J_{y_i} \cap P(\Gamma_{i-1})| + |J''_{y_i}|) + |J_{y_i} \setminus J''_{y_i}| - |J_{x_i} \setminus P(\Gamma_{i-1})| + |J'_{x_i}| \hspace*{16mm} \text{from } (\ref{eq-5}) \\
&= |P(\Gamma_{i-1})| - |J_{y_i} \cap P(\Gamma_{i-1})| + |J_{y_i}| - |J_{x_i} \setminus P(\Gamma_{i-1})| + |J'_{x_i}| \\
&\leq |P(\Gamma_{i-1})| - |J_{y_i} \cap P(\Gamma_{i-1})| + |J_{y_i} \cap P(\Gamma_{i-1})| + |J'_{x_i}| \hspace*{53mm} \text{from } (\ref{eq-4}) \\
&= |P(\Gamma_{i-1})| + |J'_{x_i}| \leq |P(\Gamma_{i-1})| + |J_{x_i} \setminus P(\Gamma_{i-1})| \hspace*{45mm} \text{by defintion of } J'_{x_i} \\
&= P(\Gamma_{i})
\end{align*}
Therefore, by the property of induction, the lemma holds.
\end{proof}

\begin{theorem}
Given the hypergraph instance $H(D,R,E)$.
Algorithm~3 computes a solution $\Gamma$ to $H$ such that $2|P(\Gamma)| \geq |P(OPT)|$, where $OPT$ is an optimal solution, with running time $O(|D| \cdot |E|)$, $|E| \leq |D| \cdot (|R|+1)^K$.
\label{theorem-ImpGreedy}
\end{theorem}

\begin{proof}
Let $\Gamma = \{x_1,\ldots,x_a\}$, $P(OPT(a)) = \bigcup_{1 \leq i \leq a} J_{y_i}$ and $P(OPT'(a)) = \bigcup_{1 \leq i \leq a} J'_{x_i}$.
By Proposition~\ref{prop-opt-size}, $P(OPT) = P(OPT(a)) \cup P(OPT'(a))$, and by Lemma~\ref{lemma-max-gap}, $|P(OPT(a))| + |P(OPT'(a))| - |P(\Gamma_{a})| \leq |P(\Gamma_a)|$.
We have
\[
|P(OPT)| \leq |P(OPT(a))| + |P(OPT'(a))| \leq 2|P(\Gamma)|.
\]
In each iteration of the while-loop, it takes $O(E)$ to find an edge $x$ with maximum $|J_x \setminus P(\Gamma)|$,
and there are at most $|D|$ iterations. Hence, Algorithm~3 runs in $O(|D| \cdot |E|)$ time.
\end{proof}

\subsection{Approximation algorithms for maximum weighted set packing}\label{sec-app-algs}
Now, we explain the algorithms for the maximum weighted set packing problem, which solve our maximization problem.
Given a universe $\mathcal{U}$ and a family $\mathcal{S}$ of subsets of $\mathcal{U}$, a \emph{packing} is a subfamily $\mathcal{C} \subseteq \mathcal{S}$ of sets such that all sets in $\mathcal{C}$ are pairwise disjoint.
Every subset $S \in \mathcal{S}$ has at most $k$ elements and is given a real weight.
The maximum weighted $k$-set packing problem (MWSP) asks to find a packing $\mathcal{C}$ with the largest total weight.
We can see that the maximization problem on $H(D,R,E)$ is a special case of the maximum weighted $k$-set packing problem, where the trips of 
$D \cup R$ is the universe $\mathcal{U}$ and $E(H)$ is the family $\mathcal{S}$ of subsets, and every $e \in E(H)$ represents at most $k = K+1$ trips ($K$ is the maximum capacity of all vehicles).
Hence, solving MWSP also solves our maximization problem.
Chandra and Halld\'{o}rsson~\cite{SODA99-C} presented a $\frac{2(k+1)}{3}$-approximation and a $\frac{2(2k+1)}{5}$-approximation algorithms (refer to as \textit{BestImp} and \textit{AnyImp} respectively), and
Berman~\cite{SWAT00-B} presented a $(\frac{k+1}{2} + \epsilon)$-approximation algorithm (refer to as \textit{SquareImp}) for the weighted $k$-set packing problem (here, $k = K + 1$), where the latter still has the best approximation ratio.

The three algorithms in~\cite{SWAT00-B,SODA99-C} (AnyImp, BestImp and SquareImp) solve the weighted $k$-set packing problem by first transferring it into a weighted independent set problem, which consists of a vertex weighted graph $G(V,E)$ and asks to find a maximum weighted independent set in $G(V,E)$.
We briefly describe the common local search approach used in these three approximation algorithms.
A \emph{claw} $C$ in $G$ is defined as an induced connected subgraph that consists of an independent set $T_C$ of vertices (called talons) and a center vertex $C_z$ that is connected to all the talons ($C$ is an induced star with center $C_z$).
For any vertex $v \in V(G)$, let $N(v)$ denotes the set of vertices in $G$ adjacent to $v$, called the \emph{neighborhood} of $v$.
For a set $U$ of vertices, $N(U) = \cup_{v \in U} N(v)$.
The \textit{local search} of AnyImp, BestImp and SquareImp uses the same central idea, summarized as follows:
\begin{enumerate}
\item The approximation algorithms start with an initial solution (independent set) $I$ in $G$ found by a \textbf{simple greedy} (refer to as \textit{Greedy}) as follows: select a vertex $u \in V(G)$ with largest weight and add to $I$.
Eliminate $u$ and all $u$'s neighbors from being selected. Repeatedly select the largest weight vertex until all vertices are eliminated from $G$.
\item While there exists claw $C$ in $G$ w.r.t. $I$ such that independent set $T_C$ improves the weight of $I$ (different for each algorithm),
augment $I$ as $I = (I \setminus N(T_C)) \cup T_C$; such an independent set $T_C$ is called an \emph{improvement}.
\end{enumerate}
To apply these algorithms to our maximization problem, we need to convert the bipartite hypergraph $H(D,R,E)$ to a weighted independent set instance $G(V,E)$, which is straightforward.
Each hyperedge $e \in E(H)$ is represented by a vertex $v_e \in V(G)$. The weight $w(v_e) = p(e)$ for each $e \in E(H)$ and $v_e \in V(G)$. There is an edge between $v_{e}, v_{e'} \in V(G)$ if $e \cap e' \neq \emptyset$ where $e, e' \in E(H)$.
We observed the following property.
\begin{property}
When the size of each set in the set packing problem is at most $k$ $(|e| = k, e \in E(H))$, the graph $G(V,E)$ has the property that it is $(k+1)$-claw free, that is, $G(V,E)$ does not contain an independent set of size $k+1$ in the neighborhood of any vertex.
\end{property}
Applying this property, we only need to search a claw $C$ consists of at most $k$ talons, which upper bounds the running time for finding a claw within $O(n^k)$, where $n = |V(G)|$.
When $k$ is very small, it is practical enough for solving our maximization problem instance $H(D,R,E)$ computed by Algorithm~2 from $(N,\mathcal{A},T)$.
It has been mentioned in~\cite{PNAS14-S} that the approximation algorithm in~\cite{SODA99-C} can be applied to the ridesharing problem. However, only the simple greedy (\textit{Greedy}) with $k$-approximation was implemented in~\cite{PNAS14-S}.
Notice that algorithm ImpGreedy (Algorithm 3) is a simplified version of algorithm Greedy, and Greedy is used to get an initial solution in algorithms AnyImp, BestImp and SquareImp. From Theorem~\ref{theorem-ImpGreedy}, we have Corollary~\ref{corollary-approximate}.
\begin{corollary}
Greedy, AnyImp, BestImp and SquareImp algorithms compute a solution to $H(D,R,E)$ with 2-approximation ratio.
\label{corollary-approximate}
\end{corollary}
Since ImpGreedy finds a solution directly on $H(D,R,E)$ without converting it to $G(V,E)$ and solving the independent set problem of $G(V,E)$, it is more time and space efficient than the algorithms for MWSP.
In the rest of this paper, Algorithm 3 is referred to as ImpGreedy.

\section{Numerical experiments} \label{sec-experiment}
We create a simulation environment consists of a centralized system that integrates public transit and ridesharing.
We implement our proposed approximation algorithm (ImpGreedy) and Greedy, AnyImp and BestImp algorithms for the $k$-set packing problem to evaluate the benefits of having an integrated transportation system supporting public transit and ridesharing.
The exact algorithm, ILP (\ref{obj-1})-(\ref{constraint-2}), is not evaluated because it takes too long to complete for the instances in our study.
The results of SquareImp are not discussed because its performance is same as AnyImp; this is due to the implementation of the search/enumeration order of the vertices and edges in the independent set instance $G(V,E)$ being fixed, and each vertex in $V(G)$ has integer weight.
We use a simplified transit network of Chicago to simulate the public transit and ridesharing.

\subsection{Description and characteristics of datasets}
We built a simplified transit network of Chicago to simulate practical scenarios of public transit and ridesharing.
The roadmap data of Chicago is retrieved from OpenStreetMap\footnote{Planet OSM. \url{https://planet.osm.org}}.
We used the GraphHopper\footnote{GraphHopper 1.0. \url{https://www.graphhopper.com}} library to construct the logical graph data structure of the roadmap.
The Chicago city is divided into 77 officially community areas, each of which is assigned an area code.
We examined two different dataset in Chicago to reveal some basic traffic pattern (the datasets are provided by the Chicago Data Portal (CDP) and Chicago Transit Authority (CTA)\footnote{CDP. \url{https://data.cityofchicago.org}. CTA. \url{https://www.transitchicago.com}}, maintained by the City of Chicago).
The first dataset is bus and rail ridership, which shows the monthly averages and monthly totals for all CTA bus routes and train station entries. We denote this dataset as \textit{PTR, public transit ridership}.
The PTR dataset range is chosen from June 1st, 2019 to June 30th, 2019.
The second dataset is rideshare trips reported by Transportation Network Providers (sometimes called rideshare companies) to the City of Chicago. We denote this dataset as \textit{TNP}.
The TNP dataset range is chosen from June 3rd, 2019 to June 30th, 2019, total of 4 weeks of data.
Table~\ref{table-PTRdata} and Table~\ref{table-TNPdata} show some basic stats of both datasets.
\begin{table}[htbp]
\small
\captionsetup{font=small}
\parbox{.5\linewidth}{
\centering
\begin{tabular}{| p{3.7cm} | p{3.3cm} |}
\hline
Total Bus Ridership & 20,300,416  \\
Total Rail Ridership & 19,282,992 \\ \hline
12 busiest bus routes & 3, 4, 8, 9, 22, 49, 53, 66, 77, 79, 82, 151 \\ \hline
The busiest bus routes selected & 4, 9, 49, 53, 77, 79, 82 \\
\hline
\end{tabular}
\caption{Basic stats of the PTR dataset \label{table-PTRdata}}
}
\hfill
\parbox{.49\linewidth}{
\centering
\begin{tabular}{| p{4.3cm} | p{2.8cm} |}
\hline
\# of original records & 8,820,037 \\ \hline 
\# of records considered & 7,427,716 \\ 
\# of shared trips & 1,015,329 \\ 
\# of non-shared trips & 6,412,387 \\ \hline
The most visited community areas selected & 1, 4, 5, 7, 22, 23, 25, 32, 41, 64, 76 \\
\hline
\end{tabular}
\caption{Basic stats of the TNP dataset \label{table-TNPdata}}
}
\end{table}

In the PTR dataset, the total ridership for each bus route is recorded; there are 127 bus routes in the dataset.
We examined the 12 busiest bus routes based on the total ridership and 
selected 7 out of the 12 routes as listed in Table~\ref{table-PTRdata} to build the transit network (excluded bus routes either serve a small community or too close to train stations). 
We also selected all the major trains/metro lines within the Chicago area except the Brown Line and Purple Line since they are too close to the Red and Blue lines.
Note that the PTR dataset also provides the total rail ridership. However, it only provides the number of riders entering every station in each day; it does not provide the number of riders exit from a station nor the time related to the entries.

Each record in the TNP dataset describes a passenger trip served by a driver who provides the rideshare service;
a trip record consists of the pick-up and drop-off time and the pick-up and drop-off community area of the trip, and exact locations are provided sometimes.
We removed records where the pick-up or drop-off community area is hidden for privacy reason or not within Chicago, which results in 7.4 million ridesharing trips.
We calculated the average number of trips per day departed from/arrived at each area.
\begin{figure}[b]
\centering
\includegraphics[width=\textwidth]{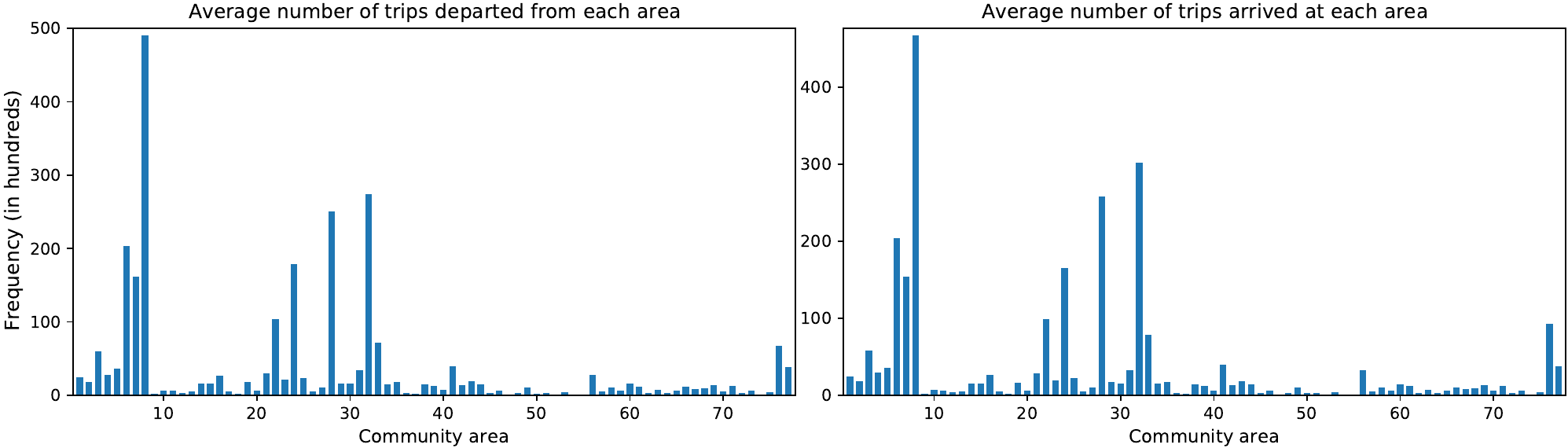}
\captionsetup{font=small}
\caption{The average number of trips per day departed from and arrived at each area.}
\label{fig-OD-pairs}
\end{figure}
The results are plotted in Figure~\ref{fig-OD-pairs}; the community areas that have the highest number of departure trips are almost the same as that of the arrival trips.

We selected 11 of the 20 most visited areas as listed in Table~\ref{table-TNPdata} (area 32 is Chicago downtown, areas 64 and 76 are airports) to build the transit network for our simulation.
From the selected bus routes, trains and community areas, we create a simplified public transit network connecting urban community areas, depicted in Figure~\ref{fig-transit-network}.
\begin{figure}[!tp]
\centering
\includegraphics[width=\textwidth]{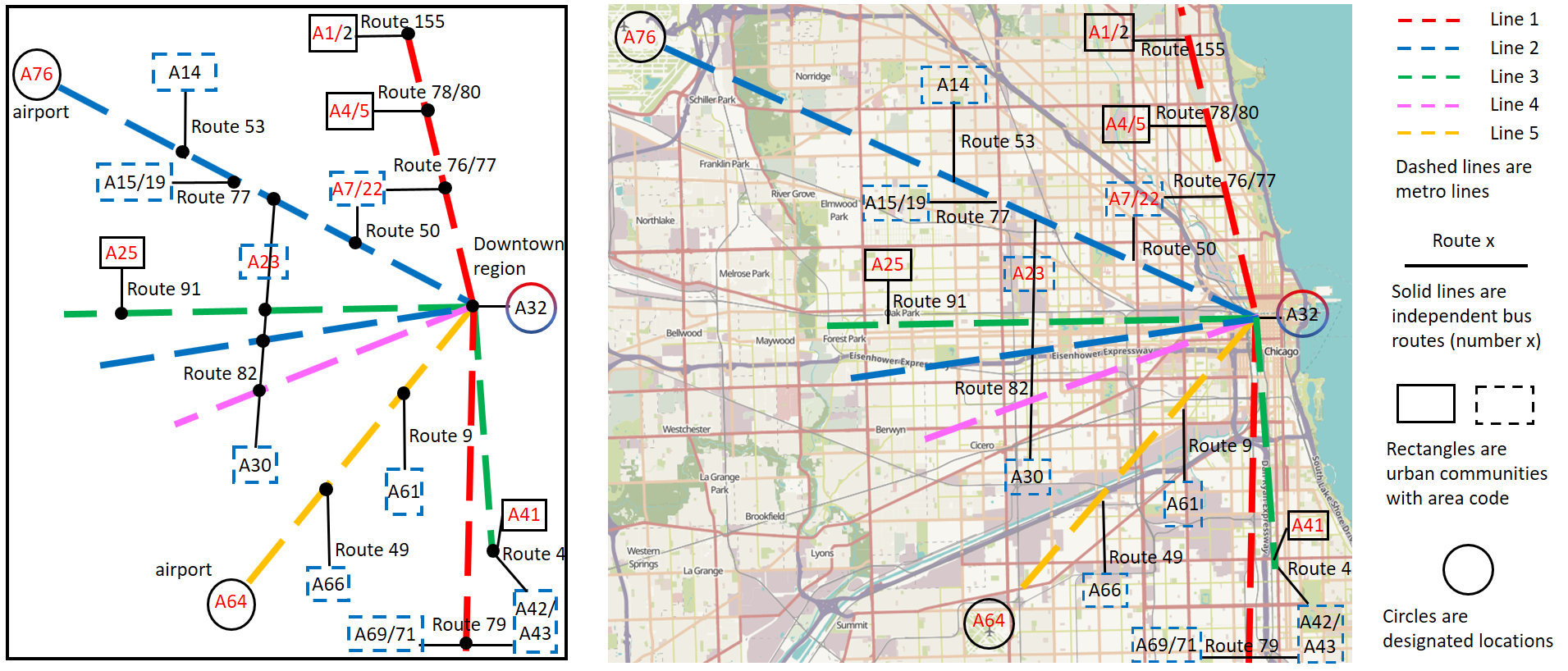}
\captionsetup{font=small}
\caption{Simplified public transit network of Chicago with 13 urban communities and 3 designated locations. Figure on the right has the Chicago city map overlay for scale.}
\label{fig-transit-network}
\end{figure}
Each rectangle on the figure represents an urban community within one community area or across two community areas, labeled in the rectangle.
The blue dashed rectangles/urban communities are chosen due to the busiest bus routes from the PTR dataset.
The rectangles/urban communities labeled with red area codes are chosen due to the most visited community areas from the TNP dataset.
The dashed lines are the trains, which resemble the major train services in Chicago. The solid lines are the selected bus routes connecting the urban communities to their closest train stations.
There are also three designated locations/destinations that many people want to travel to/from throughout the day; they are the two airports and downtown region in Chicago.

The travel time between two locations (each location consists of the latitude and longitude coordinates) uses the fastest/shortest route computed by the GraphHopper library, which is based on personal cars.
The shortest paths are \textbf{computed in real-time}, unlike many previous simulations where the shortest paths are precomputed and stored.
As mentioned in Section~\ref{subsection-alg-feas-single}, waiting time and service time are considered in a simplified model; we multiply a small constant $\epsilon > 1$ to the fastest route to mimic waiting time and service time for public transit.
For instance, consider two consecutive metro stations $s_1$ and $s_2$. The travel time $t(s_1,s_2)$ is computed by the fastest route, and the travel time by train between from $s_1$ to $s_2$ is $\hat{t}(s_1,s_2) = 1.15 \cdot t(s_1,s_2)$. The constant $\epsilon$ for bus service is 2.
Rider trips originated from most locations must take a bus to reach a metro station when ridesharing service is not involved.

\subsection{Generating instances}\label{sec-instances}
In our simulation, we partition each day from 6:00 to 23:59 into 72 time intervals (each has 15 minutes), and we only focus on weekdays.
To see ridesharing traffic pattern, we calculated the average number of served trips per hour for each day of the week using the TNP dataset.
The dashed (orange) line and solid (blue) line of the plot in Figure~(\ref{fig-sub-originalTrips}) represent shared trips and non-shared trips respectively.
A set of trips are called \emph{shared trips} if this set of trips are matched for the same vehicle consecutively such that their trips may potentially overlap, namely, one or more passengers are in the same vehicle.
For all other trips, we call them \textit{non-shared trips}.
From the plot, the peak hours are between 7:00 AM to 9:00AM and 4:00PM to 7:00PM on weekdays for both non-shared and shared trips.
The number of trips generated for each interval is plotted in Figure~(\ref{fig-sub-nTrips}), which is a scaled down and smoothed version of the TNP dataset for weekdays.
The ratio between the number of drivers and riders generated is roughly 1:3 (1 driver and 3 riders) for each interval.
Such a ratio is chosen because it should reflect the system's potential as capacity of 3 is common for most vehicles.
\begin{figure}[!ht]
\captionsetup{font=small}
\centering
\begin{subfigure}{.62\textwidth}
  \centering
  \includegraphics[width=1.0\linewidth]{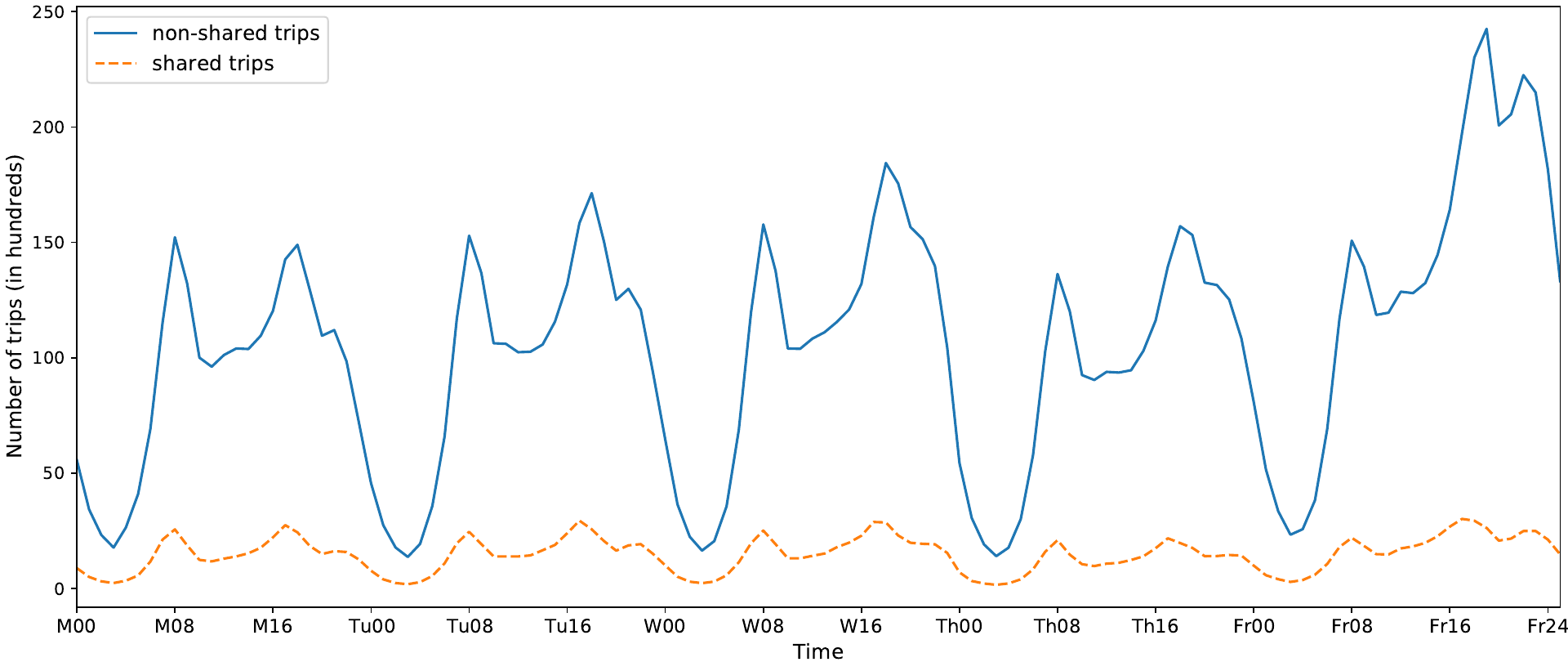}
  \caption{Average numbers of shared and non-shared trips in TNP dataset.}
  \label{fig-sub-originalTrips}
\end{subfigure}%
\hfill
\begin{subfigure}{.37\textwidth}
  \centering
  \includegraphics[width=0.98\linewidth]{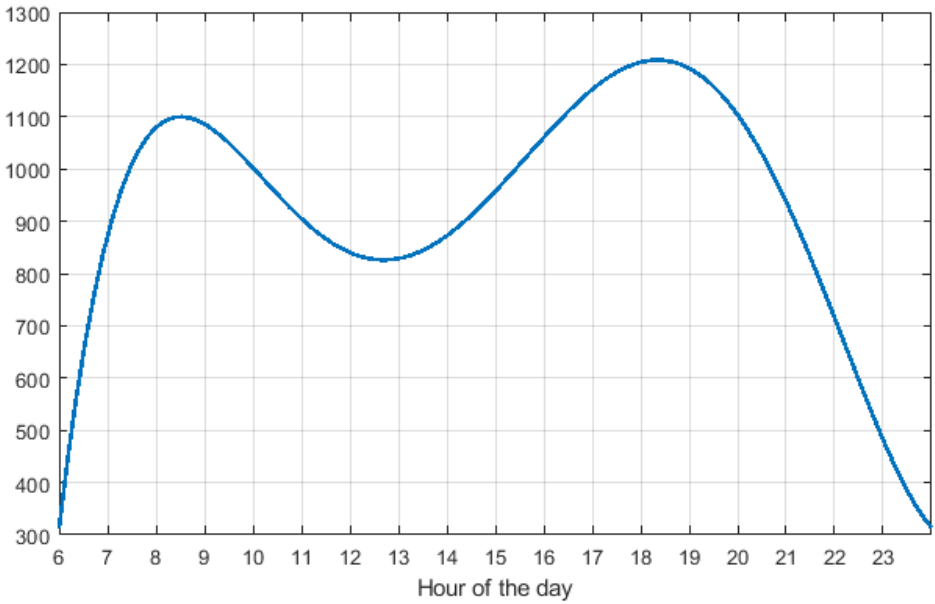}
  \caption{Total number of driver and rider trips generated for each time interval.}
  \label{fig-sub-nTrips}
\end{subfigure}
\caption{Plots for the number of trips for every hour from data and generated.}
\label{fig-nTrips-plot}
\vspace*{-1mm}
\end{figure}
For each time interval, we first generate a set $R$ of riders and then a set $D$ of drivers.
We do not generate a trip where its origin and destination are close. For example, no trip with origin Area25 and destination Area15 is generated.

\paragraph{Generation of rider trips.}
We assume that the numbers of riders entering and exiting a station are the same.
Next we assume that the the numbers of riders in PTR over the time intervals each day follow a similar distribution of the TNP trips over the time intervals.
Each day is divided into 6 different consecutive time periods (each consists of multiple time intervals):
morning rush, morning normal, noon, afternoon normal, afternoon rush, and evening time periods.
Each time period determines the probability and distribution of origins and destinations.
Based on the PTR dataset and Rail Capacity Study by CTA~\cite{CTA19}, many riders are going into downtown in the morning and leaving downtown in the afternoon.
To generate a rider trip $j$ during \textbf{morning rush} time period, we first decide a \emph{pickup area} which is a community area selected uniformly at random. The origin $o_j$ is a random point within the selected pickup area.
Then, we use standard normal distribution to determine the \emph{dropoff area}, where downtown area is within two SDs (standard deviations), airports are more than two and at most three SDs, and the community areas are more than three SDs away from the mean.
The destination $d_j$ is a random point within the selected dropoff area.
The above is repeated until $a_t$ riders are generated, where $a_t + a_t / 3$ (riders + drivers) is the total number of trips for time interval $t$ shown in Figure~(\ref{fig-sub-nTrips}).
For any pickup area $c$, let $c_t$ be the number of generated riders originated from $c$ for time interval $t$, that is, $\sum_c c_t = a_t$.
Other time periods follow the same procedure, and all community areas and locations can be selected as pickup and dropoff areas
\begin{enumerate}
\setlength\itemsep{0em}
\item \textbf{Morning normal}: for pickup area, community areas are within two SDs, downtown is more than two and at most three SDs and airports are more than three SDs away from the mean; and destination area is selected using uniform distribution.
\item \textbf{Noon}: both pickup and dropoff are selected using uniform distribution.
\item \textbf{Afternoon normal}: for pickup area, downtown and airport are within two SDs and community areas are more than two SDs away from the mean; for dropoff area, community areas are within two SDs and downtown and airports are more than two SDs away from the mean.
\item \textbf{Afternoon rush}: for pickup area, downtown is within two SDs, airports are more than two SDs and at most three SDs and community areas are more than three SDs away from the mean; and for dropoff area, community areas are within two SDs, airports are more than two SDs and at most three SDs and downtown is more than three SDs away from the mean.
\item \textbf{Evening}: for both pickup and dropoff areas, community areas are within two SDs, downtown is more than two and at most three SDs and airports are more than three SDs away from the mean.
\end{enumerate}

\paragraph{Generation of driver trips.}
We examined the TNP dataset to determine if there are enough drivers who can provide ridesharing service to riders that follow match Types 1 and 2 traffic pattern.
First, we removed any trip from TNP if it is too short (less than 15 minutes or origin and destination are adjacent areas).
We calculated the average number of trips per hour originated from every pre-defined area in the transit network (Figure~\ref{fig-transit-network}), and then plotted the destinations of such trips in a grid heatmap.
In other words, each cell $(c,r)$ in the heatmap represents the the average number of trips per hour originated from area $c$ to destination area $r$ in the transit network (Figure~\ref{fig-transit-network}).
An example is depicted in Figure~\ref{fig-OD-distribution}.
\begin{figure}[!htbp]
\centering
\includegraphics[width=\textwidth]{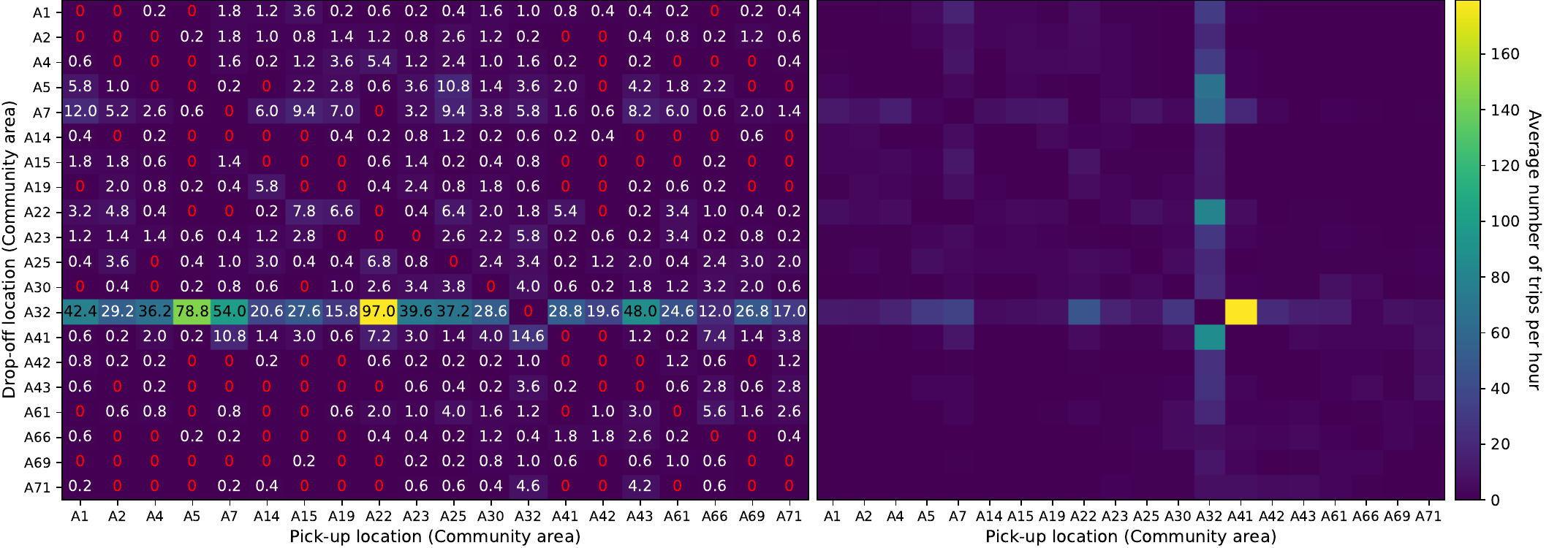}
\caption{Traffic heatmaps for the average number of trips originated from one area (x-axis) during hour 7:00 (left) and hour 17:00 (right) to every other destination area (y-axis).}
\label{fig-OD-distribution}
\end{figure}
From the heatmaps, many trips are going into the downtown area (A32) in the morning; and as time progresses, more and more trips leave downtown. This traffic pattern confirms that there are enough drivers to serve the riders in our simulation.
The number of shared trips shown in Figure~\ref{fig-sub-originalTrips} also suggests that many riders are willing to share a same vehicle.
We slightly reduce the difference between the values of each cell in the heatmaps and use the idea of marginal probability to generate driver trips.
Let $d(c,r,h)$ be the value at the cell $(c,r)$ for origin area $c$, destination $r$ and hour $h$.
Let $P(c,h)$ be sum of the average number of trips originated from area $c$ for hour $h$ (the column for area $c$ in the heatmap corresponds to hour $h$), that is, $P(c, h) = \sum_{r} d(c,r,h)$ is the sum of the values of the whole column $c$ for hour $h$.
Given a time interval $t$, for each area $c$, we generate $c_t/3$ drivers ($c_t$ is defined in Generation of rider trips) such that each driver $i$ has origin $o_i = c$ and destination $d_i = r$ with probability $d(c,r,h)/P(c,h)$, where $t$ is contained in hour $h$.
The probability of selecting an airport as destination is fixed at 5\%.

\paragraph{Deciding other parameters for each trip.}
After the origin and destination of a rider of driver trip have been determined, we decide other parameters of the trip.
The capacity $n_i$ of drivers' vehicles is selected from three ranges: the {\em low range} [1,2,3], {\em mid range} [3,4,5], and {\em high range} [4,5,6].
During morning/afternoon peak hours, roughly 95\% and 5\% of vehicles have capacities randomly selected from the low range and mid range respectively.
It is realistic to assume vehicle capacity is lower for morning and afternoon peak-hour commute.
While during off-peak hours, roughly 80\%, 10\% and 10\% of vehicles have capacities randomly selected from low range, mid range and high range respectively.
The number $\delta_i$ of stops equals to $n_i$ if $n_i \leq 3$, else it is chosen uniformly at random from $[n_i-2, n_i]$ inclusive.
The detour limit $z_i$ of each driver is within 5 to 20 minutes because traffic is not considered, and waiting time and service time are considered in a simplified model.
The general information of the base instances is summarized in Table~\ref{table-simulation}.
\begin{table}[htbp]
\footnotesize
\centering
   \begin{tabular}{ l | p{11.3cm} }
   	\hline
   	Major trip patterns       & from urban communities to downtown and vice versa for peak and off-peak hours respectively;
                                          trips specify one match type for peak hours and can be in either type for off-peak hours \\
    \# of intervals simulated & Start from 6:00 AM to 11:59 PM; each interval is 15 minutes		\\
   	\# of trips per interval  & varies from [350, 1150] roughly, see Figure~\ref{fig-nTrips-plot}  	\\ 
   	Driver:rider ratio        & 1:3 approximately												\\ 
   	Capacity $n_i$ of vehicles 	  & low: [1,3], mid: [3,5] and high: [4,6] inclusive \\ 
   	Number $\delta_i$ of stops limit 	& $\delta_i=n_i$ if $n_i \leq 3$, or $\delta_i \in [n_i-2,n_i]$ if $n_i \geq 4$		\\
    Earliest departure time $\alpha_i$    & immediate to 2 intervals after a trip announcement is generated  \\
    Driver detour limit $z_i$      & 5 minutes to min\{$2 \cdot t(o_i,d_i)$ (driver's fastest route), 20 minutes\}	\\
   	Latest arrival time $\beta_i$  & at most $1.5 \cdot (t(o_i,d_i) + z_i) + \alpha_i$                       \\
    Travel duration $\gamma_i$ of driver $i$  & $\gamma_i = t(o_i,d_i) + z_i$   		\\
    Travel duration $\gamma_j$ of rider $j$  & $\gamma_j = t(\hat{\pi}_i)$, where $\hat{\pi}_i$ is the fastest public transit route    		\\ 
    Acceptance rate      & 80\% for all riders	(0.8 times the fastest public transit route)           	\\   
    Train and bus travel time   & average at 1.15 and 2 times the fastest route by car, respectively \\ \hline
   \end{tabular}
\captionsetup{font=small}
\caption{General information of the base instances.}
\label{table-simulation}
\end{table}

\paragraph{Reduction configuration procedure.}
When the number of trips increases, the running time for Algorithm~2 and the time needed to construct the $k$-set packing instance also increase. This is due to the increased number of feasible matches for each driver $i \in D$.
In a practical setup, we may restrict the number of feasible matches a driver can have.
Each match produced by Algorithm~1 is called a \emph{base match}, which consists of exactly one driver and one passenger.
To make the simulation feasible, we heuristically limit the numbers of base matches for each driver and each rider and the number of total feasible matches for each driver. We use $(x\%, y, z)$, called \emph{reduction configuration} (\emph{Config} for short), to denote that for each driver $i$, the number of base matches of $i$ is reduced to $x$ percentage and at most $y$ total feasible matches are computed for $i$; and for each rider $j$, at most $z$ base matches containing $j$ are used.

After Algorithm~1 is completed. A reduction procedure may be evoked with respect to a reduction Config.
Let $H(D,R,E)$ be the graph after computing all feasible base matches (instance computed by Algorithm~1 and before Algorithm~2 is executed).
For a trip $i \in \mathcal{A}$, let $E_i$ be the set of base matches of $i$.
The reduction procedure works as follows.
\begin{itemize}
\setlength\itemsep{0em}
\item First of all, the set of drivers is sorted, based on number of base matches each driver has, in descending order.
\item Each driver $i$ is then processed one by one.
    \begin{enumerate}
    \item If driver $i$ has at least 10 base matches, then $E_i$ is sorted, based on the number of base matches each passenger included in $E_i$ has, in descending order.
    \item For each match $e=(i, J=\{j\})$ in $E_i$, if $j$ belongs to more than $z$ other matches, remove $e$ from $E_i$.
    \item After above step 2, if $E_i$ has not been reduced to $x\%$, sort the remaining matches of $E_i$, based on the travel time of passengers  included in $E_i$ to $i$, in descending order.
    \item Remove the first $x'$ matches from $E_i$ until $x\%$ is reached.
    \end{enumerate}
\end{itemize}
The original sorting of the drivers allows us to first remove matches from drivers that have more matches than others.
The sorting of the base matches of driver $i$ in step 1 allows us to first remove matches containing passengers that also belong to other matches.
Passengers farther away from a driver $i$ may have lower chance to be served together by $i$; this is the reason for the sorting in step 3.

\subsection{Computational results}
We use the same transit network and same set of generated trip data for all algorithms.
All experiments were implemented in Java and conducted on Intel Core i7-2600 processor with 1333 MHz of 8 GB RAM available to JVM.
Since the optimization goal is to assign accepted ridesharing route to as many riders as possible, the performance measure is focused on the number of riders served by ridesharing routes, followed by the total time saved for the riders as a whole.
We record both of these numbers for each approximation algorithm.
The base case instance uses the parameter setting described in Section~\ref{sec-instances} and Config (30\%, 600, 20).
The experiment results are shown in Table~\ref{table-base-result}.
\begin{table}[htbp]
\footnotesize
\centering
   \begin{tabular}{ l | c | c | c | c }
   	\hline
                                                                    & ImpGreedy 	& Greedy       & AnyImp       & BestImp     \\ \hline
    Total number of riders served                    & 27413  	        & 27413         & 28248          & 28258         \\    
    Avg number of riders served per interval    & 380.736  	    & 380.736      & 392.333       & 392.472       \\
   Total time saved of all riders (minute)	            & 354568.2  	    & 354568.2    & 365860.6     & 365945.8     \\
    Avg time saved of riders per interval (minute)	& 4924.56        & 4924.56      & 5,081.40      & 5082.58   \\ \hline
    \multicolumn{2}{| l |}{Total number of riders and public transit duration} & \multicolumn{3}{l |}{45314 and 1383743.97 minutes}\\ \hline
   \end{tabular}
\captionsetup{font=small}
\caption{Base case solution comparison between the approximation algorithms.}
\label{table-base-result}
\end{table}
The results of ImpGreedy and Greedy are aligned since they are essentially the same algorithm - 60.5\% of total passengers are assigned ridesharing routes and 25.6\% of total time are saved.
The results of AnyImp and BestImp are similar because of the density of the graph $G(V,E)$ due to Observation~\ref{obs-1}.
For AnyImp and BestImp, roughly 62.4\% of total passengers are assigned ridesharing routes and 26.4\% of total time are saved.
On average, passengers are able to reduce their travel duration from 30.5 minutes to 22.5 minutes by using public transit plus ridesharing.
The results of these four algorithms are not too far apart.
However, it takes too long for AnyImp and BestImp to run to completion.
A 10-second limit is set for both algorithms in each iteration for finding an independent set improvement.
With this time limit, AnyImp and BestImp run to completion within 15 minutes for almost all intervals.

We also examine this from the drivers' perspective; we recorded both the mean occupancy rate and vacancy rate of drivers.
The mean occupancy rate is calculated as, in each interval, the number of passengers served divided by the number of drivers who serve them.
The mean vacancy rate is calculated as, in each interval, the number of drivers with feasible matches who are not assigned any passenger divided by the total number of drivers with at least one feasible match.
The results are depicted in Figure~\ref{fig-OR-VR}.
\begin{figure}[!ht]
\centering
\includegraphics[width=\textwidth]{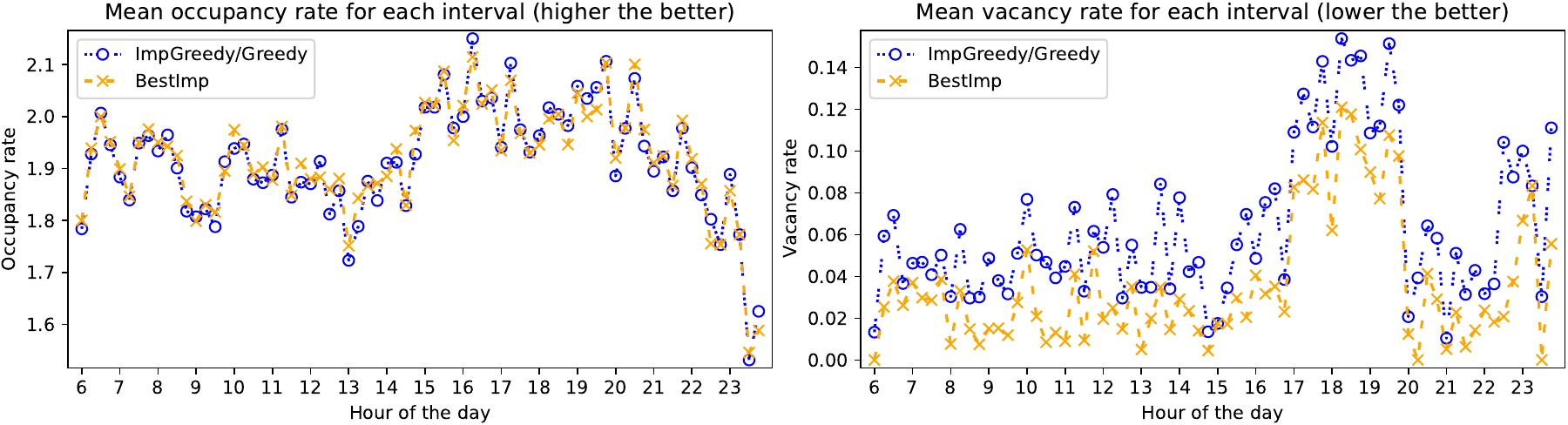}
\captionsetup{font=small}
\caption{The mean occupancy rate and vacancy rate of drivers for each interval.}
\label{fig-OR-VR}
\end{figure}
The occupancy rate results show that in many intervals, 1.9-2 passengers are served by each driver on average.
The vacancy rate of drivers show that 3-8\% (0-4\% resp.) of drivers are not assigned any passenger while such drivers have some feasible matches  for ImpGreedy (BestImp respectively) during all hours except afternoon peak hours; on the other hand, this time period has the highest occupancy rate.
This is most likely due to the origins of many trips are from the same area (downtown). If the destinations of drivers and riders do not have the same general direction from downtown, the drivers may not be able to serve any riders. On the other hand, when their destinations are aligned, drivers are likely to serve more riders.

Another major component of the experiment is to measure the computational time of the algorithms, which is highly affected by the base match reduction configurations.
By reducing more matches, we are able to improve the running time of AnyImp and BestImp significantly, but sacrifice performance slightly.
We tested 12 different Configs:
\begin{itemize}[leftmargin=*]
\setlength\itemsep{0em}
\begin{footnotesize}
\item \textit{Small1} (20\%,300,10), \textit{Small2} (20\%,600,10), \textit{Small3} (20\%,300,20), \textit{Small4-10} (20\%,600,20).
\item \textit{Medium1} (30\%,300,10), \textit{Medium2} (30\%,600,10), \textit{Medium3} (30\%,300,20), \textit{Medium4-10} (30\%,600,20).
\item\textit{Large1} (40\%,300,10), \textit{Large2} (40\%,600,10), \textit{Large3-10} (40\%,300,20), and \textit{Large4-10} (40\%,600,20).
\end{footnotesize}
\end{itemize}
Configs with label ``-10'' have a 10-second limit to find an independent set improvement, and all other Configs have 20-second limit.
Notice that all 12 Configs have the same sets of driver/rider trips and base match sets but generate different feasible match sets.
The performance and running time results of all 12 Configs are depicted in Figures~\ref{fig-configurations-perf}~and~\ref{fig-configurations-time} respectively.
The results are divided into peak and off-peak hours for each Config (averaging all intervals of peak hours and off-peak hours).
The running time of ImpGreedy and Greedy are within seconds for all Configs as shown in Figure~\ref{fig-configurations-time}.
On the other hand, it may not be practical to use AnyImp and BestImp for peak hours since they require around 15 minutes for most Configs.
Since AnyImp and BestImp provide better performance than ImpGreedy/Greedy when each Config is compared side-by-side, one can use ImpGreedy/Greedy for peak hours and AnyImp/BestImp for off-peak hours so that it becomes practical.
\begin{figure}[!t]
\centering
\includegraphics[width=\textwidth]{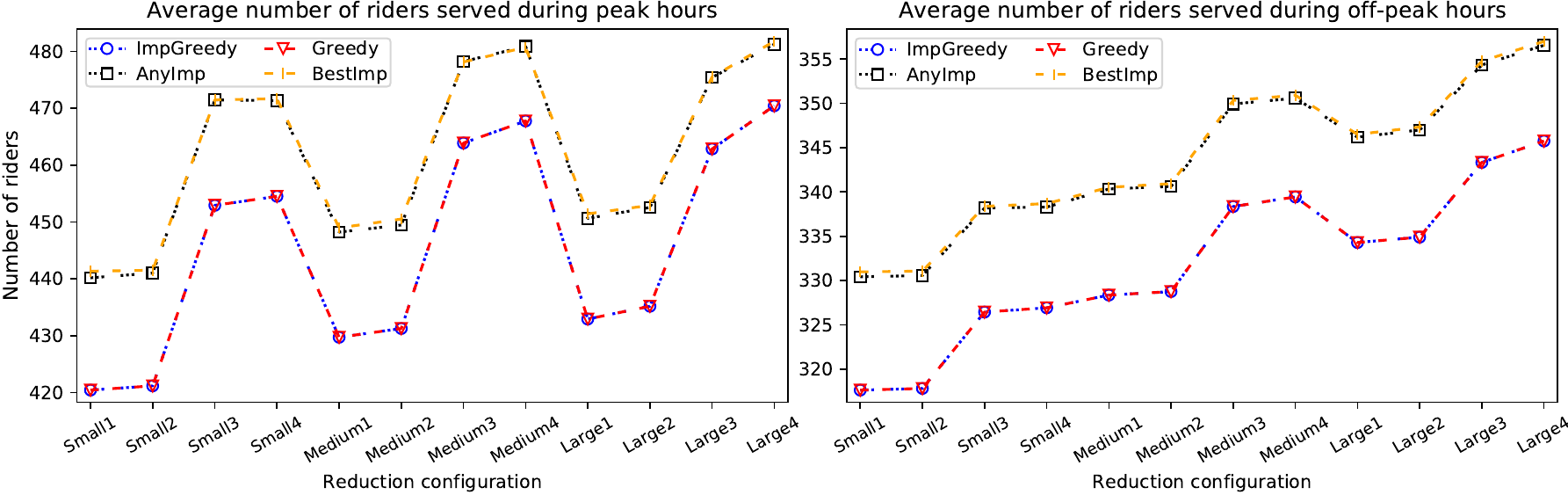}
\captionsetup{font=small}
\caption{Average performance of peak and off-peak hours for different configurations.}
\label{fig-configurations-perf}
\end{figure}
\begin{figure}[!t]
\centering
\includegraphics[width=\textwidth]{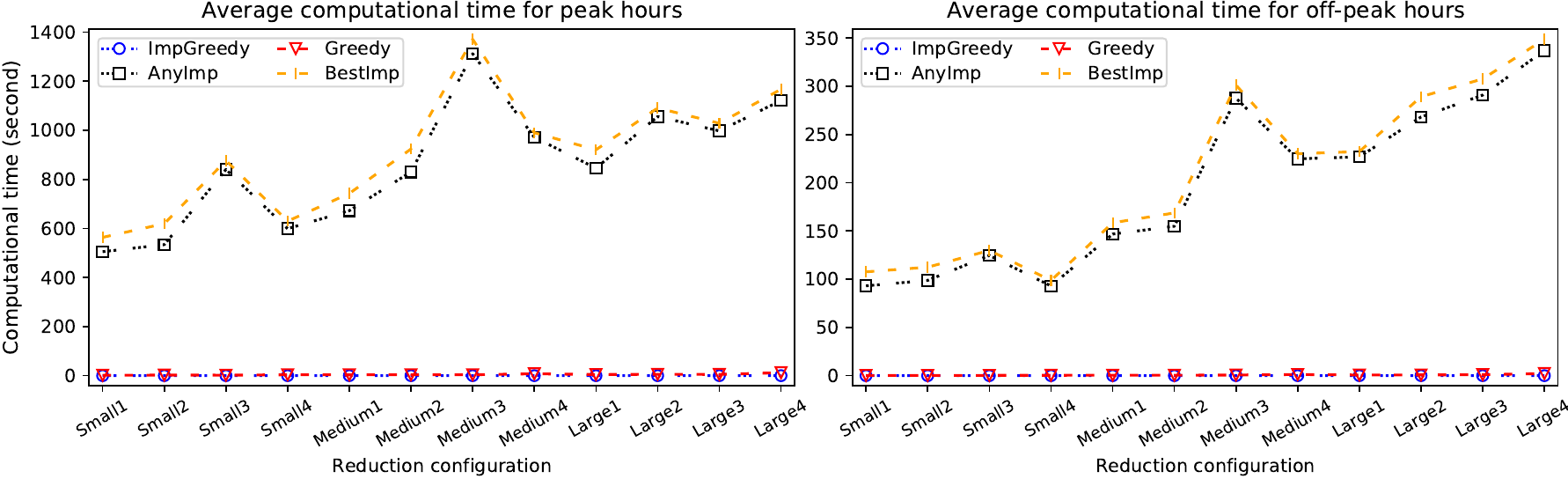}
\captionsetup{font=small}
\caption{Average running time of peak and off-peak hours for different configurations.}
\label{fig-configurations-time}
\end{figure}
The increase in performance from Small1 to Small3 is much larger than that from Small1 to Small2 (same for Medium and Large), implying any parameter in a Config should not be too small.
The increase in performance from Large1 to Large4 is higher than that from Medium1 to Medium4 (similarly for Small).
Therefore, a balanced configuration is more important than a configuration emphasizes only one or two parameters.

Because ImpGreedy does not create the independent set instance, it runs quicker than Greedy. More importantly, ImpGreedy uses less memory space than Greedy does.
We tested ImpGreedy and Greedy with the following Configs: \textit{Huge1} (100\%,600,10), \textit{Huge2} (100\%,2500,20) and \textit{Huge3} (100\%,10000,30)  (these Configs have the same sets of driver/rider trips and base match sets as those in the previous 12 Configs). The focus of these Configs is to see if Greedy can handle large number of feasible matches.
The results are shown in Table~\ref{table-hugeConfig}.
\begin{table}[htbp!]
\footnotesize
\centering
\begin{tabular}{ l | c | c | c }
\hline
	\textbf{ImpGreedy}                                                       & Huge1            & Huge2           & Huge3        \\ \hline
	Avg running time for peak/off-peak hours (sec)               & 0.08 / 0.03      & 0.43 / 0.12     & 1.2 / 0.29        \\
	Avg number of riders served for peak/off-peak hours      & 406.9 / 339.0  & 458.8 / 355.4  & 484.1 / 361.9  \\
	Avg time saved of riders per interval (sec)                       & 284891.8       & 302774.1         & 310636.9   \\ \hline
	\textbf{Greedy}                                                             & Huge1           & Huge2             & Huge3       \\ \hline
	Avg running time                                                            & N/A               & N/A                & N/A         \\
	Avg instance size $G(V,E)$ of afternoon peak ($|E(G)|$) & 0.02 billion     & 0.38 billion      & 5.47 billion \\
	Avg time creating $G(V,E)$ of afternoon peak (sec)        & 14.6               & 320.9              & 3726.79  \\ \hline
\end{tabular}
\captionsetup{font=small}
\caption{The results of ImpGreedy and Greedy using Unlimited reduction configurations.}
\label{table-hugeConfig}
\end{table}
Greedy cannot run to completion for all configurations because in many intervals, the whole graph $G(V,E)$ of the independent set instance is too large to hold in memory.
The average number of edges for afternoon peak hours is 0.02, 0.38 and 5.47 billion for Huge1, Huge2 and Huge3 respectively.
Further, the time it takes to create $G(V,E)$ can excess practicality.
Hence, using Greedy (AnyImp and BestImp) for large instances may not be practical.
In addition, the performance of ImpGreedy with Huge3 is better than that of AnyImp/BestImp with Large4.

Lastly, we also looked at the total running times of the approximation algorithms including the time for computing feasible matches (Algorithms 1 and 2). The running time of Algorithm 1 solely depends on computing the shortest paths between the trips and stations.
Table~\ref{table-algorithms-time} shows that Algorithm 1 runs to completion within 500 seconds on average for peak hours.
As for Algorithm 2, when many trips' origins/destinations are concentrated in one area, the running time increases significantly, especially for drivers with high capacity.
Running time of Algorithm 2 can be reduced significantly by Configs with aggressive reductions.
\begin{table}[htbp]
\scriptsize
\setlength\tabcolsep{5pt} 
\centering
   \begin{tabular}{  l | c | c | c | c | c | c || c|c|c|c }
              & Alg1     & Alg2     & ImpGreedy & Greedy  & AnyImp & BestImp & \multicolumn{4}{c}{Total computational time} \\
              & & & & & &                                                                                 & ImpGreedy & Greedy & AnyImp & BestImp  \\ \hline
   {Small3}      & 485.2   & 26.8     & 0.021    & 2.0      & 840.5    & 876.4      & 512.1  & 514.1  & 1352.5  & 1388.5 \\
   {Small4}      & 485.2   & 28.2     & 0.029    & 3.6      & 599.1    & 629.9      & 513.4  & 517.0  & 1112.5  & 1143.3 \\
   {Medium3}  & 485.2   & 43.6     & 0.031    & 3.7      & 1312.1  & 1371.0    & 532.5  & 543.0  & 1840.9  & 1899.9 \\
   {Medium4}  & 485.2   & 50.1     & 0.048    & 7.7      & 971.5    & 990.0      & 535.3  & 543.0  & 1506.8  & 1525.3 \\
   {Large4}      & 485.2   & 72.0     & 0.076   & 12.2     & 1121.3  & 1167.2     & 557.3  & 569.5  & 1678.6  & 1724.4 \\
   {Huge3}      & 485.2   & 339.4    & 1.2       & N/A     & N/A      & N/A         & 825.8  & N/A    & N/A      & N/A \\
   \end{tabular}
\captionsetup{font=small}
\caption{Average computational time (in seconds) of peak hours for all algorithms.}
\label{table-algorithms-time}
\end{table}
Combining the results of this and previous (Table~\ref{table-hugeConfig}) experiments, ImpGreedy is capable of handling large instances while providing quality solution compared to other approximation algorithms.

From the experiment results in Figures~\ref{fig-configurations-perf}~and~\ref{fig-configurations-time}, it is beneficial to dynamically select different algorithms and reduction configurations for each interval depending on the number of trips.
With large problem instances, previous approximation algorithms are not efficient (time and memory consuming), so they require aggressive reduction to reduce the instance size.
On the other hand, ImpGreedy is much faster and capable of handling large instances.
The running time of ImpGreedy can also be an advantage to improve the quality of solutions. For example, as shown in Figures~\ref{fig-configurations-perf}~and~\ref{fig-configurations-time}, for the same set of drivers and riders, ImpGreedy assigns more riders when taking Meduim/Medium4 as inputs than AnyImp/BestImp on Small1/Small2, and uses less time than AnyImp/BestImp.
When the size of an instance is not small and a solution must be computed within some time-limit, ImpGreedy has a distinct advantage over the previous approximation algorithms.


\section{Conclusion and future work} \label{sec-conclusion}
Based on real-world transit datasets in Chicago, our study has shown that integrating public and private transportation can benefit the transit system as a whole, 
Recall that we focus on work commute traffic, and we only consider two match types that emphasize this transit pattern (with the flexibility to choose either type).
Just from these two types, our base case experiments show that more than 60\% of the passenger are assigned ridesharing routes and able to save 25\% of travel time.
Majority of the drivers are matched with at least one passenger, and vehicle occupancy rate has improved close to 3 (including the driver) on average.
These results suggest that ridesharing can be a complement to public transit.
Our experiments show that the whole system is capable of handling more than 1000 trip requests in real-time using ordinary computer hardware.
It is likely that the performance results of ImpGreedy can be further improved by extending it with the local search strategy of AnyImp and BestImp.
Perhaps the biggest challenge for scalability comes from computing the base matches (Algorithm 1) since it has to compute many shortest paths in real-time;
it may be worth to apply heuristics to reduce the running time of Algorithm 1 for scalability.
To better understand scalability and practicality, it is important to include different match types and a more sophisticated simulation which includes real transit schedule and transit demand.

\end{document}